\documentclass[
	a4paper,
	11pt,
	DIV=9,
]
{scrartcl}

\usepackage[utf8]{inputenc}
\usepackage[watermark,showeqnr,theoremdefs,final]{latexdev}
\usepackage{klas}
\usepackage{framed}

\usepackage[numbers]{natbib} 

\usepackage{authblk} 

\providecommand{\iftoggle}[3]{#2}

\title{Generalised Hunter--Saxton equations, \\ optimal information transport, \\ and factorisation of diffeomorphisms}
\author{Klas Modin${}^{1,2,}$\thanks{\href{mailto:klas.modin@utoronto.ca}{\tt klas.modin@chalmers.se}}}
\affil{
	${}^1$Department of Mathematics \\
	University of Toronto \\
	Toronto, ON M5S 2E4, Canada
}
\affil{
	${}^2$Mathematical Sciences \\
	Chalmers University of Technology \\
	SE--412 96 Göteborg, Sweden
}

\date{\Large\today \\[-0.5ex] {\footnotesize (Published in J.~Geom.~Anal., \href{http://dx.doi.org/10.1007/s12220-014-9469-2}{DOI:10.1007/s12220-014-9469-2})}}

	\providecommand{\mathup}{\mathrm}

	\newcommand{\R}{{\mathbb R}}
	\newcommand{\C}{{\mathbb C}}

	\newcommand{\pd}{\partial}
	
	\newcommand{\dd}{\mathup{d}}
	\newcommand{\ud}{\,\dd}

	\DeclareMathOperator{\grad}{grad}
	\DeclareMathOperator{\divv}{div}

	\DeclareMathOperator{\Jac}{Jac}
	
	\newcommand{\LieD}{\mathcal{L}}
	\newcommand{\interior}{\mathup{i}}
	
	\newcommand{\vol}{\mathup{vol}}
	\providecommand{\met}{\mathsf{g}}
	\providecommand{\dist}{\mathup{dist}}
	
	\DeclareMathOperator{\Exp}{Exp}

	\providecommand{\pair}[1]{\langle #1 \rangle}
	\providecommand{\inner}[1]{\langle\!\langle #1 \rangle\!\rangle}
	\DeclareMathOperator{\tr}{tr}
	\providecommand{\id}{\mathup{id}}
	
	\newcommand{\trans}{\top}

	\DeclareMathOperator{\ad}{ad}

	\newcommand{\g}{\mathfrak{g}}
	
	\providecommand{\GL}{\mathup{GL}}
	\providecommand{\gl}{\mathfrak{gl}}

	\providecommand{\SO}{\mathup{SO}}
	\providecommand{\so}{\mathfrak{so}}
	
	\newcommand{\Fcal}{\mathcal{F}}
	\newcommand{\Diff}{\mathup{Diff}}
	\newcommand{\Xcal}{\mathfrak{X}}
	\newcommand{\Diffvol}{{\Diff_\vol}}
	\newcommand{\Xcalvol}{{\Xcal_\vol}}
	
	\newcommand{\Xcalvolex}{{\Xcal_{\vol,\mathrm{ex}}}}
	
	\newcommand{\Xcalhar}{{\Xcal_{\mathcal{H}}}}

	\newcommand{\Har}{\mathcal{H}}
	
	\newcommand{\Dens}{\mathup{Dens}}

	\newcommand{\Symn}{\mathup{Sym}(n)}
	
	\newcommand{\Sympn}{\mathup{Sym}(n)^+}
	\newcommand{\Matn}{\mathup{Mat}(n,n)}

	\newcommand{\Acal}{\mathcal{A}}

	\providecommand{\todo}[1]{}
	
\begin{document}

\maketitle

\begin{abstract}
	We study geodesic equations for a family of right-invariant Riemannian metrics on the group of diffeomorphisms of a compact manifold.
	The metrics descend to Fisher's information metric on the space of smooth probability densities. 
	The right reduced geodesic equations are higher-dimensional generalisations of the $\mu$--Hunter--Saxton equation, used to model liquid crystals under influence of magnetic fields.
	Local existence and uniqueness results are established by proving smoothness of the geodesic spray.

	The descending property of the metrics is used to obtain a novel factorisation of diffeomorphisms.
	Analogous to the polar factorisation in optimal mass transport, this factorisation solves an optimal information transport problem.
	It can be seen as an infinite-dimensional version of $QR$~factorisation of matrices.
		
	\textbf{Keywords:} Euler--Arnold equations; Euler--Poincare equations; descending metrics; Riemannian submersion; diffeomorphism groups; Fisher information metric; Fisher--Rao metric; entropy differential metric; geometric statistics; Hunter--Saxton equation; information geometry; optimal transport; polar factorisation; $QR$~factorisation; Cholesky factorisation; Calabi metric.

	\textbf{MSC~2010:} 58D05, 58D15, 35Q31, 53C21, 58B20, 94A17, 65F99.
\end{abstract}

\listoftodos

\tableofcontents

\section{Introduction} 
\label{sec:intro}

Let $(M,\met)$ be a closed $n$--dimensional Riemannian manifold.
Denote by $\vol$ the volume form.
We assume that 
$\int_M \vol = 1$.
Denote by $\Diff(M)$ the group of diffeomorphisms of~$M$ and by $\Diffvol(M)$ the subgroup of volume preserving diffeomorphisms.
Throughout the paper, the word ``metric'' always means ``Riemannian metric''.

The study of geodesic equations on diffeomorphism groups was initiated by \makebox{\citet{Ar1966}}, who discovered that Euler's equations for an incompressible perfect fluid correspond to a geodesic equation on $\Diffvol(M)$ with respect to a right-invariant $L^2$~metric. 
Since Arnold's discovery, many equations in mathematical physics have been found to fit the same framework.
Such equations are called \emph{Euler--Arnold equations};
see the monographs~\cite{ArKh1998, MaRa1999, KhWe2009, HoScSt2009} and the survey paper~\cite{Vi2008}.

This paper is about geodesic equations for a family of right-invariant metrics on $\Diff(M)$ that descend to the homogeneous space $\Diffvol(M)\backslash\Diff(M)$ of right co-sets, naturally identified with the space $\Dens(M)$ of smooth probability densities.

Riemannian metrics and geodesic equations on $\Dens(M)$ are important in optimal transport, probability theory, statistical mechanics, and quantum mechanics.
The connection between geodesics on $\Diff(M)$ and $\Dens(M)$ was pointed out by \citet{Ot2001}, who studied a non-invariant $L^{2}$~metric on $\Diff(M)$ that descends to $\Dens(M)$.
(In Otto's setting, $\Dens(M)$ is identified with the homogeneous space $\Diff(M)/\Diffvol(M)$ of \emph{left} co-sets.)
Remarkably, the corresponding metric on $\Dens(M)$ induces the Wasserstein distance; it is therefore called the \emph{Wasserstein metric}.
Otto's observation implies that~$L^{2}$~optimal mass transport can be interpreted (at least formally) as a geodesic boundary-value problem on $\Dens(M)$ with respect to the Wasserstein metric.
This follows from general result for Riemannian submersions: a minimal geodesic between two fibres is horizontal, and horizontal geodesics descend to geodesics on the base.

Another important metric on the space of probability densities is the \emph{Fisher metric} (also called \emph{Fisher--Rao metric}, \emph{Fisher information metric}, and \emph{entropy differential metric}).
Classically, the Fisher metric occurs as a finite-dimensional metric on smooth statistical models (called statistical manifolds). 
It has a fundamental rôle in the field of information geometry~\cite{Fi1973,Ra1945,Ch1982,AmNa2000}. 
\citet{Fr1991} realised that statistical manifolds can be interpreted as finite-dimensional submanifolds of $\Dens(M)$ and that the Fisher metric on such submanifolds is the restriction of one and the same canonical metric on $\Dens(M)$.
Friedrich also showed that this canonical metric has constant positive curvature.

When $(M,J,\omega)$ is a complex manifold, the Calabi--Yau theorem establishes an isomorphism between $\Dens(M)$ and the space of Kähler metrics on $M$ compatible with the symplectic structure $\omega$.
Under this isomorphism, the Fisher metric corresponds to the \emph{Calabi metric} introduced in the 1950s (see \cite{ClRu2013} for details).
Results about the Fisher metric are therefore of interest also in the field of Kähler--Einstein metrics and complex Monge--Ampere equations.

\citet*{KhLeMiPr2013} introduced a right-invariant degenerate~$\dot H^{1}$ ``metric'' on $\Diff(M)$ that descends to the Fisher metric on~$\Dens(M)$.
($\Dens(M)$ is here identified with the \emph{right} co-sets $\Diffvol(M)\backslash\Diff(M)$.)
By taking Otto's point of view, the authors regard the geodesic boundary-value problem, with respect to the Fisher metric, as an optimal ``information'' transport problem, with respect to a degenerate cost function on $\Diff(M)$ induced by the $\dot H^{1}$~``metric''.
Since the cost function is degenerate, the transport maps are not unique.

As pointed out in \cite{KhLeMiPr2013}, there are no available examples of right-invariant metrics on $\Diff(M)$ that descend to $\Dens(M)$.
Our main motivation is to construct such metrics, thereby completing the analogy between optimal mass and information transport.
Indeed, in this paper we introduce a 3--parameter family of right-invariant metrics on $\Diff(M)$ that descend to the Fisher metric.
We give existence and uniqueness results for the geodesic equation and the corresponding optimal information transport problem.
The latter result implies a novel factorisation of diffeomorphisms.
This factorisation is analogous to, but different from, the polar factorisation of vector valued maps on~$\R^{n}$, obtained by \makebox{\citet{Br1991}} and later generalised to Riemannian manifolds by \citet{Mc2001}.
Our factorisation result can be understood as an infinite-dimensional version of the $QR$~factorisation of matrices.

The right reduced geodesic equations for our family of metrics can be interpreted as higher-dimensional generalisations of the $\mu$--Hunter--Saxton ($\mu$HS) equation, studied by \citet*{KhLeMi2008} (also called $\mu$--Camassa--Holm in~\cite{LeMiTi2010}).
The $\mu$HS~equation is a simple model for a liquid crystal under influence of an external magnetic field. 

We now present the higher-dimensional generalisations.

Let $\Xcal(M)$ and $\Omega^{k}(M)$, respectively, denote smooth vector fields and $k$--forms on $M$.
Further, let
\begin{equation*}
	\Fcal(M) = \left\{ F \in C^{\infty}(M); \int_{M} F\, \vol = 0 \right\} .
\end{equation*}
Recall the differential $\dd\colon\Omega^k(M)\to\Omega^{k+1}(M)$ and the co-differential $\delta\colon\Omega^k(M)\to\Omega^{k-1}(M)$.
The Laplace--de~Rham operator $\Delta \coloneqq -\dd\circ\delta - \delta\circ\dd$ restricted to $\dd\Omega^{k-1}(M)$ or $\delta\Omega^{k+1}(M)$ is an isomorphism~\cite{Ta1996a}.
In particular, it is an isomorphism on $\Fcal(M) = \delta\Omega^{1}(M)$.
Let $\flat\colon\Xcal(M)\to\Omega^1(M)$ denote the flat map, also called the \emph{musical isomorphism}.
Its inverse, the sharp map, is denoted~$\sharp$.
We typically write $u^\flat$ instead of $\flat(u)$ and correspondingly for~$\sharp$.

Consider the pseudo-differential operator $\Acal\colon\Xcal(M)\to\Omega^1(M)$ defined by
\begin{equation}\label{eq:Acal}
	\Acal u \coloneqq 
	\Big( 
		\id + \dd\circ \Delta^{-1}\circ\delta + \gamma\,\delta\circ\Delta^{-1}\circ\dd + \alpha\,\delta\circ\dd + \beta\,\dd\circ\delta
	\Big) (u^{\flat})
\end{equation}
where $\alpha,\beta >0$ and $\gamma \in [0,1]$ are parameters.
We are interested in the integro-differential equation
\begin{subequations}\label{eq:maineq}
\begin{equation}\label{eq:maineq1}
	\dot m + \LieD_u m + m \divv (u) = 0, \quad m =\Acal u ,
\end{equation}
where $\LieD_u$ denotes the Lie derivative along $u$ and $\dot m \coloneqq \frac{\pd m}{\pd t}$.
A \emph{solution} is a curve $t\mapsto u(t)\in\Xcal(M)$ that fulfils equation~\eqref{eq:maineq1}.
The equation also admits the form
\begin{equation}\label{eq:maineq2}
	\Big( \frac{\pd}{\pd t} + \LieD_u \Big) (m\otimes\vol) = 0 ,
\end{equation}
because
\begin{equation*}
	\LieD_u (m\otimes\vol) = (\LieD_u m)\otimes \vol + m\otimes \divv(u)\vol = \big(\LieD_u m + m\divv(u)\big)\otimes\vol .
\end{equation*}
The one-form density $m\otimes\vol$ is therefore transported by the flow, much like vorticity is transported by the flow of a perfect fluid.

\begin{remark}
	If $\alpha=\beta=\gamma=1$ and the first de~Rham cohomology of $M$ is trivial, then $\Acal = -\Delta\circ\flat$.
	An example is $M = S^{n}$ (the $n$--dimensional sphere) for $n>1$.
\end{remark}	

\begin{remark}
	If $\alpha=\beta=\gamma=1$ and $M=\mathbb{T}^{n}$ (the $n$--dimensional flat torus) then
	\begin{equation*}
		\Acal u = -\Delta u^{\flat} + \sum_{i=1}^{n} \pair{u,t_i}_{L^{2}}\, t_i^{\flat},
	\end{equation*} 
	where $t_i \coloneqq \frac{\pd}{\pd{x_i}} \in \Xcal(\mathbb{T}^{n})$, i.e., $t_1,\ldots,t_n$ is a basis for infinitesimal translations on $\mathbb{T}^{n}$.
\end{remark}

The paper is organised as follows.
In \autoref{sec:euler_arnold} we show that equation~\eqref{eq:maineq} is a right reduced geodesic equation on $\Diff(M)$, i.e., an Euler--Arnold equation.
Local existence and uniqueness of the Cauchy problem is given in~\autoref{sec:cauchy_problem}.
In \autoref{sec:descending_metrics} we discuss characterisation and construction of right-invariant and descending metrics, and we show that our family of constructed metrics descend to the Fisher metric.
In \autoref{sec:optimal_transport_and_polar_factorisation} we present an abstract geometric framework for right-invariant optimal transport problems and polar factorisation.
Then, in \autoref{sub:optimal_info_trans}, we focus on optimal information transport, using as cost function the squared Riemannian distance of the new metrics, and we derive a polar factorisation result for $H^{s}$~diffeomorphisms.
Finally, we show in \autoref{sub:qr} that $QR$~factorisation of matrices can be viewed as polar factorisation corresponding to optimal transport of inner products on~$\R^{n}$.
The relation to the Cholesky factorisation of symmetric matrices is pointed out.

We continue the introduction by deriving yet another form of equation~\eqref{eq:maineq}.
This form reveals structural properties and relations to other equations.

\subsection{Hodge components} 
\label{sub:hodge_components}


The Helmholtz decomposition of vector fields is
\begin{equation*}
\Xcal(M) = \Xcalvol(M)\oplus \grad(\Fcal(M)),
\end{equation*}
where $\Xcalvol(M)$ is the space of divergence-free vector fields;
every $u \in \Xcal(M)$ can be decomposed uniquely as $u = \xi + \grad(f)$, with $\xi\in\Xcalvol(M)$ and $f \in \Fcal(M)$.
Since~$f$ is normalised, it is unique.
This decomposition is orthogonal with respect to the $L^{2}$~inner product on $\Xcal(M)$, given by
\begin{equation*}
	\pair{u,v}_{L^{2}} = \int_{M}\met(u,v)\vol .
\end{equation*}

The Hodge decomposition of $k$--forms is
\begin{equation*}
	\Omega^k(M) = \Har^{k}(M)\oplus\delta\Omega^{k+1}(M)\oplus\dd\Omega^{k-1}(M), 	
\end{equation*}
where $\Har^k(M) = \{ a \in \Omega^k(M); \Delta a=0 \}$ is the space of harmonic $k$--forms.
This decomposition is orthogonal with respect to the $L^{2}$~inner product on $\Omega^{k}(M)$, given by
\begin{equation*}
	\pair{a,b}_{L^{2}} = \int_{M} a\wedge\star b ,
\end{equation*}
where $\star:\Omega^{k}(M)\to\Omega^{n-k}(M)$ is the Hodge star map.
Notice that $\pair{u,v}_{L^{2}} = \pair{u^{\flat},v^{\flat}}_{L^{2}}$.

Let $\mathsf{D}^k(M) \coloneqq \Har^{k}(M)\oplus\delta\Omega^{k+1}(M)$.
Then $\mathsf{D}^k(M) = \{ a \in \Omega^k(M); \delta a = 0 \}$ is the space of co-closed $k$--forms.
The relation between the Helmholtz and Hodge decompositions is
\begin{equation*}
	\Xcalvol(M)^{\flat} = \mathsf{D}^{1}(M),\quad \grad(\Fcal(M))^{\flat} = \dd\Omega^{0}(M).
\end{equation*} 
That is, the musical isomorphism $\flat\colon\Xcal(M)\to \Omega^{1}(M)$ is diagonal with respect to the two decompositions.
The same is true for the pseudo-differential operator $\Acal\colon\Xcal(M)\to\Omega^1(M)$.
That is, 
\[
	\Acal\Xcalvol(M) = \mathsf{D}^1(M) ,
	\quad
	\Acal\grad(\Fcal(M)) = \dd\Omega^0(M).
\]
From the Hodge decomposition we also obtain a finer decomposition
\[
\Xcalvol(M) = \Xcalhar(M)\oplus\Xcalvolex(M) ,
\] 
where $\Xcalvolex(M) = \delta\Omega^{2}(M)^\sharp$ is the space of \emph{exact} volume preserving vector fields, and $\Xcalhar(M) = \Har^{1}(M)^\sharp$ is the space of \emph{harmonic} vector fields.
$\Acal$~is diagonal also with respect to this finer decomposition.
Indeed, the $L^2$~orthogonal projection operator~$R\colon\Omega^{1}(M)\to\Har^{1}(M)$ onto the harmonic part is given by
\[
	R = \id + \dd\circ\Delta^{-1}\circ\delta + \delta\circ\Delta^{-1}\circ\dd ,
\]
and the $L^2$~orthogonal projection operator~$P\colon\Omega^{1}(M)\to\mathsf{D}^{1}(M)$ onto the co-closed part is given by
\[
	P = \id + \dd\circ\Delta^{-1}\circ\delta .
\]
From the definition of $\Acal$ it follows that $\Acal = (\gamma R + (1-\gamma) P + \alpha\,\delta\circ\dd + \beta\,\dd\circ\delta)\circ\flat$.
If $h\in\Xcalhar(M)$ then
\[
	\Acal{h} = \underbrace{\big(\gamma R + (1-\gamma)P\big){h}^{\flat}}_{{h}^{\flat}}+ \alpha\delta\underbrace{\dd{h}^{\flat}}_{0}+\beta\dd\underbrace{\delta{h}^{\flat}}_{0} = {h}^{\flat} \in \Har^{1}(M).
\]
If $\xi\in\Xcalvolex(M)$ then
\[
	\Acal\xi = \underbrace{\gamma R\xi^{\flat}}_{0}+ (1-\gamma)\underbrace{P\xi^{\flat}}_{\xi^{\flat}} + \alpha\delta\dd\xi^{\flat}+\beta\dd\underbrace{\delta\xi^{\flat}}_{0} = (1-\gamma)\xi^{\flat} + \alpha\delta\dd\xi^{\flat} \in \delta\Omega^{2}(M).
\]
If $f\in\Fcal(M)$ then
\[
	\Acal\grad(f) = \underbrace{(1-\gamma)P\dd f + \gamma R\dd f + \alpha\, \delta\dd\dd f}_{0} - \beta\,\dd\Delta f = -\beta\,\dd\Delta f \in \dd\Omega^0(M).
\]
Thus, if we represent $u = h + \xi + \grad(f)$ by its unique ``Helmholtz--Hodge components'' $(h,\xi,f)\in\Xcalhar(M)\times\Xcalvolex(M)\times\Fcal(M)$, then 
\[
	\Acal(h,\xi,f) = (h^{\flat},\big((1-\gamma)\id-\alpha\Delta\big)\xi^\flat,-\beta\Delta f)\in\Har^{1}(M)\times\delta\Omega^{2}(M)\times\Fcal(M).
\]
Since both $\big((1-\gamma)\id - \alpha\Delta\big)\circ\flat\colon\Xcalvolex(M)\to\delta\Omega^{2}(M)$ and $\Delta\colon\Fcal(M)\to\Fcal(M)$ are invertible operators, $\Acal$ is also invertible (see \autoref{sec:cauchy_problem} for details).

Our aim is now to write equation~\eqref{eq:maineq1} in terms of the Hodge components
\[
\sigma \coloneqq (\gamma R + (1-\gamma)P + \alpha\,\delta\circ\dd)(u^{\flat})\in\mathsf{D}^{1}(M)
\]
\todo{Referee's remark.}
and
\[
	\rho \coloneqq \Delta f =\divv(u)\in\Fcal(M).
\]
In these variables $m=\sigma - \beta\dd\rho$, so equation~\eqref{eq:maineq1} becomes
\begin{gather*}
	\dot\sigma - \beta\,\dd\dot\rho 
	+ \LieD_{u}\sigma - \beta\,\dd\LieD_u\rho 
	+ \rho\sigma - \beta \rho\,\dd\rho = 0
	\\
	\Updownarrow
	\\
	\dot\sigma 
	+ \LieD_u\sigma
	+ \rho \sigma 
	- \beta\,\dd \Big( 
		\dot\rho + \LieD_u\rho + \frac{\rho^2}{2}
	\Big) = 0
\end{gather*}
In general, $\LieD_u\sigma + \rho \sigma\notin \mathsf{D}^1(M)$ and $\LieD_u\rho + \frac{\rho^2}{2}\notin\Fcal(M)$.
Therefore we need a Lagrange multiplier in order to find the Hodge components.
We can always find a function~$p\in C^\infty(M)$ such that $\LieD_u\xi^\flat + \rho \xi^\flat + \dd p \in \mathsf{D}^1(M)$, with $p$ uniquely determined up to a constant.
Further, we can always determine the constant part of~$p$ so that $\LieD_u\rho + \frac{\rho^2}{2} + \frac{p}{\beta} \in\Fcal(M)$.
Continuing from above
\begin{gather*}
	\Updownarrow
	\\
	\underbrace{\dot\sigma 
	+ \LieD_u\sigma
	+ \rho \sigma 
	+ \dd p}_{\in\mathsf{D}^1(M)}
	- \beta\,\dd \Big( 
		\underbrace{\dot\rho + \LieD_u\rho + \frac{\rho^2}{2} + \frac{p}{\beta}}_{\in\Fcal(M)}
	\Big) = 0. 
\end{gather*}
We now obtain equation~\eqref{eq:maineq1} in terms of the Hodge components
\begin{equation}\label{eq:maineq3}
	\begin{split}
		\dot\sigma + \LieD_u \sigma + \rho\sigma &= -\dd p, \qquad \sigma = \big(\gamma R + (1-\gamma)\id - \alpha\Delta\big)(Pu^{\flat}) \\
		\dot\rho + \LieD_{u}\rho + \frac{\rho^2}{2} &= - \frac{p}{\beta}, \qquad \rho = \divv(u) \\
		\delta\sigma &= 0  \\
		\int_{M} \rho \, \vol &= 0 ,
	\end{split}
\end{equation}
\end{subequations}
where the ``pressure'' $p\in C^{\infty}(M)$ is a Lagrange multiplier, determined uniquely by the two constraint equations.
 
If $\sigma(t_0) = 0$ at some time $t_0$, it follows from equation~\eqref{eq:maineq3} that $\dot\sigma(t_0) = 0$. 
The consequence is that $\grad(\Fcal(M))$ is an \emph{invariant subspace}; if $u(t_0)\in\grad(\Fcal(M))$ then $u(t)\in\grad(\Fcal(M))$ for all~$t$.
From a geometric point of view, the reason is that the corresponding right-invariant metric on $\Diff(M)$ descends to the homogenous space $\Diffvol(M)\backslash\Diff(M) \simeq \Dens(M)$, as described in \autoref{sec:descending_metrics}.
In contrast, $\rho(t_0) = 0$ does \emph{not} imply~$\rho(t)=0$, so $\Xcalvol(M)$ is \emph{not} an invariant subspace, so $\Diffvol(M)$ is \emph{not} totally geodesic (see~\cite{MoPeMaMc2011} for details on totally geodesic subgroups).
But, if $\rho(t_0)=0$ then it follows from equation~\eqref{eq:maineq3} that $\dot\rho(t_0)$ is arbitrarily small for large enough~$\beta$.
This observation suggests that solutions to equation~\eqref{eq:maineq} with $\gamma=0$ may converge to solutions of the Euler--$\alpha$ fluid equation as~$\beta\to\infty$.
We do not expect good behaviour of solutions as $\beta\to 0$, since~$\Acal$ is not invertible for~$\beta=0$.
 
Equation~\eqref{eq:maineq} is a higher-dimensional generalisation of the $\mu$HS equation, studied by \citet*{KhLeMi2008}.
Indeed, $\Xcal_{\vol}(S^{1}) = \Xcalhar(S^{1})\simeq \R$ consists of the constant vector fields on the circle~$S^{1}$, so equation~\eqref{eq:maineq3} with $M=S^{1}$ becomes
\begin{equation*}
	\begin{split}
		\dot\xi + 2 \xi u_{x} &= -p_{x}
		\\
		\dot u_{x} + u u_{xx} + \frac{1}{2}(u_{x})^{2} &= -\frac{p}{\beta} .
	\end{split}
\end{equation*}
From the first equation it follows that 
\begin{equation*}
	0= \int_{S^{1}}(\dot\xi + 2\xi u_{x} + p_{x})\ud x = \dot\xi \mu(S^{1}),
	\quad\mu(S^{1})\coloneqq\int_{S^{1}}\ud x,
\end{equation*}
which implies~$\dot\xi = 0$.
Now differentiate the second equation with respect to~$x$ to get
\begin{equation*}
	\dot u_{xx} + 2 u_{x} u_{xx} + u u_{xxx} = \frac{2 \xi u_{x}}{\beta}  .
\end{equation*}
Since $\int_{S^{1}}u \ud x = \int_{S^{1}}\xi \ud x$ it follows that $\xi$ is the mean of $u$ over~$S^{1}$, i.e.,
\begin{equation*}
	\xi = \mu(u) \coloneqq \frac{1}{\mu(S^{1})}\int_{S^{1}} u \ud x  .
\end{equation*}
Thus,
\begin{equation*}
	\dot u_{xx} + 2 u_{x} u_{xx} + u u_{xxx} = \frac{2 \mu(u) u_{x}}{\beta},
\end{equation*}
which is the $\mu$HS equation.

A different generalisation of the $\mu$HS equation, from $M=S^1$ to $M=\mathbb{T}^n$, is given by \citet{Ko2012}.
It is also an Euler--Arnold equation, but the corresponding right-invariant metric does not descend to density space.


\section{Euler--Arnold structure} 
\label{sec:euler_arnold}

The geodesic equation for a right-invariant (or left-invariant) metric on a Lie group~$G$ can be reduced to an equation on the Lie algebra $\g$, called an Euler--Poincaré or Euler--Arnold equation.
The abstract form of this equation, first written down by \citet{Po1901}, is
\begin{equation}\label{eq:EulerArnold_abstract}
	\Acal\dot u + \ad^{*}_{u}(\Acal u) = 0,
\end{equation}
where $\Acal\colon\g\to\g^{*}$ is the \emph{inertia operator} induced by the inner product on~$\g$ defining the right-invariant metric, and $\ad^{*}_{u}\colon\g^{*}\to\g^{*}$ is the infinitesimal action of~$u$ on $\g^{*}$, i.e., the dual operator of $\ad_{u}\colon\g\to\g$.

In our case, $G=\Diff(M)$, $\g = \Xcal(M)$, and $\ad_{u} = -\LieD_{u}$.
The dual of $\Xcal(M)$ is identified with $\Omega^{1}(M)$ via the pairing
\begin{equation*}
	\pair{m,u} = \int_{M}\interior_{u}m \, \vol = \pair{m,v^{\flat}}_{L^{2}} .
\end{equation*}
We now define an inner product on $\Xcal(M)$ with inertia operator given by~\eqref{eq:Acal}.
Indeed, consider the inner product
\begin{equation}\label{eq:inner_product}
	\begin{split}
		\pair{u,v}_{\alpha\beta\gamma} 
		&\coloneqq \pair{\bar P_{\gamma} u^{\flat},\bar P_{\gamma} v^{\flat}}_{L^{2}} +\alpha\pair{\dd u^{\flat},\dd v^{\flat}}_{L^{2}}+ \beta\pair{\delta u^{\flat},\delta v^{\flat}}_{L^{2}} ,
	\end{split}
\end{equation}
where $\bar P_{\gamma} \coloneqq \gamma R + (1-\gamma) P$ is introduced to simplify notation.
Notice that \eqref{eq:inner_product} is different from the Sobolev $a$-$b$-$c$~inner product considered in~\cite{KhLeMiPr2013}, since only the divergence-free components occur in the first term.
From $\pair{a,\delta b}_{L^{2}} = \pair{\dd a,b}_{L^{2}}$ and $\pair{\bar P_{\gamma} u^{\flat},\bar P_{\gamma} v^{\flat}}_{L^{2}} = \pair{\bar P_{\gamma} u^{\flat}, v^{\flat}}_{L^{2}}$ we get
\begin{equation*}
	\pair{u,v}_{\alpha\beta\gamma}  = \pair{\bar P_{\gamma} u^{\flat} + \alpha\, \delta\dd u^{\flat} + \beta\, \dd\delta u^{\flat},v^{\flat} }_{L^{2}} = \pair{\Acal u,v},
\end{equation*}
so $\Acal$ in \eqref{eq:Acal} is indeed the inertia tensor corresponding to the inner product~\eqref{eq:inner_product}.

Using the inner product~\eqref{eq:inner_product}, we define a right-invariant metric $\inner{\cdot,\cdot}_{\alpha\beta\gamma}$~on $\Diff(M)$ by right translation of vectors to $\Xcal(M) = T_{\id}\Diff(M)$. 
Explicitly,
\begin{equation}\label{eq:metric}
	\inner{U,V}_{\alpha\beta\gamma}
	= \pair{U\circ\varphi^{-1},V\circ\varphi^{-1}}_{\alpha\beta\gamma} \; ,
\end{equation}
for $U,V\in T_{\varphi}\Diff(M)$.

Let us revisit the case $M=S^{1}$.
$\Diffvol(S^{1}) = \mathrm{Rot}(S^{1})$, i.e., the one dimensional Lie group of rigid rotations, and $\Xcalvol(S^{1})\simeq \R$, i.e., the constant vector fields.
The inner product~\eqref{eq:inner_product} becomes
\begin{equation*}
	\pair{u,v}_{\mu\dot H^{1}} 
	= \int_{S^{1}} u \ud s \int_{S^{1}}v \ud s + \int_{S^{1}} u_{x} v_{x} \ud s,
\end{equation*}
which is the inner product for the $\mu$HS~metric in~\cite{KhLeMi2008}.
As expected, the $\alpha$-$\beta$-$\gamma$~metric~\eqref{eq:metric} on $\Diff(M)$ is thus a generalisation of the $\mu$HS~metric on $\Diff(S^{1})$.

The dual operator of $-\LieD_{u}\colon\Xcal(M)\to\Xcal(M)$ is computed as
\begin{equation*}
	\begin{split}
		\pair{m,-\LieD_{u}v} &= -\int_{M} m\wedge\star (\LieD_{u}v)^{\flat} 
		= -\int_{M}m\wedge\interior_{\LieD_{u}v}\vol
		\\ &
		= -\int_{M}m\wedge \big( \LieD_{u}\interior_{v}\vol - \interior_{v}\LieD_{u}\vol \big)		
		\\
		& = -\underbrace{\int_{M}\LieD_{u} \big( m\wedge \interior_{v}\vol \big)}_{0} + \int_{M}\LieD_{u}m\wedge \big(\interior_{v}\vol+ \divv(u)\interior_{v}\vol \big)		
		\\
		& =  \int_{M}\big(\LieD_{u}m + \divv(u)m \big)\wedge \interior_{v}\vol = \pair{\LieD_{u}m + \divv(u)m,v}.		
	\end{split}
\end{equation*}
Hence, $\ad_{u}^{*}(m)=\LieD_{u}m + \divv(u)m$. 
We obtain the following result by comparing equations~\eqref{eq:maineq} and~\eqref{eq:EulerArnold_abstract}.

\begin{proposition}\label{pro:EulerArnoldEquation}
	Equation~\eqref{eq:maineq} is the Euler--Arnold equation for the geodesic flow on $\Diff(M)$ with respect to the right-invariant $\alpha$-$\beta$-$\gamma$~metric~\eqref{eq:metric}.
\end{proposition}


\section{Local existence and uniqueness} 
\label{sec:cauchy_problem}

In this section we show that equation~\eqref{eq:maineq} is well-posed as a Cauchy problem.
Following \citet{EbMa1970}, the approach is to prove smoothness of the geodesic spray in Sobolev $H^s$~topologies.

Let $N$ be a smooth finite-dimensional manifold.
If $s>n/2$ then the set $H^{s}(M,N)$ of maps from $M$ to $N$ of Sobolev differentiability~$H^{s}$ is a Banach manifold ($H^{s}(M,N)$ is even a Hilbert manifold, but that is not relevant in our analysis).
Let $\pi_{N}\colon TN\to N$ be the canonical projection.
The tangent space at $f\in H^{s}(M,N)$ is
\begin{equation*}
	T_{f}H^{s}(M,N) = \{ v \in H^{s}(M,TN); \pi_{N}\circ v = f\} ,
\end{equation*}
so $T H^{s}(M,N) = H^{s}(M,TN)$.
By iteration we obtain the higher order tangent bundles as $T^{k}H^{s}(M,N) = H(M,T^{k}N)$.

If $s>n/2+1$, which is assumed throughout the remainder, then $\Diff^{s}(M)$, i.e., the set of bijective maps in $H^{s}(M,M)$ whose inverses also belong to $H^{s}(M,M)$, is an open subset of $H^{s}(M,M)$, and therefore also a Banach manifold.
Since $\Diff^{s}(M)$ is open in $H^{s}(M,M)$ we have $T_{\varphi}\Diff^{s}(M) = T_{\varphi}H^{s}(M,M)$.
In particular, $T_{\id}\Diff^{s}(M) = \Xcal^{s}(M)$, i.e., the vector fields on $M$ of Sobolev type~$H^{s}$.

The space of $k$--forms of Sobolev type~$H^{s}$ is denoted $\Omega^{k,s}(M)$.

Let $\psi\in\Diff^{s}(M)$. 
Right multiplication $\Diff^{s}(M)\ni\varphi\mapsto \varphi\circ\psi \in \Diff^{s}(M)$ is smooth, but $\Diff^s(M)$ is \emph{not} a Banach Lie group, because left multiplication is \emph{not} smooth.
Instead, $\Diff^s(M)$ is a topological group, because the group operations are continuous~\cite[\S\!~2]{EbMa1970}.

\iftoggle{final}{}{{\color{blue!50!black}
\begin{framed}
We now introduce a local coordinate chart in a neighbourhood of $\varphi_0\in\Diff^{s}(M)$.
Thus, let $(g,\dot g)\in (O,T_{\varphi_{0}}\Diff^{s}(M))$ denote local coordinates for $T\Diff^{s}(M)$, where $O\subset T_{\varphi_{0}}\Diff^{s}(M)$ is an open neighbourhood of zero.
Notice that
\begin{equation*}
	T_{\varphi_{0}}\Diff^{s}(M) = \{ \psi \in H^{s}(M,TM); \pi_{M}\circ\psi = \varphi_{0}\}.
\end{equation*}
The chart mapping $O\to \Diff^{s}(M)$ is given by $g \mapsto \exp_{M}\circ g$, where $\exp_{M}:TM\to M$ is the Riemannian exponential on~$M$.

If $\psi\in H^{s}(M,M)$, then a chart from a neighbourhood of zero in $T_{\psi}H^{s}(M,M)$ to a neighbourhood of $\psi\in H^{s}(M,M)$ is given by ...
\end{framed}
}}

\iftoggle{final}{}{{\color{blue!50!black}
\begin{framed}
The outline of the proof is as follows:
\begin{enumerate}
	\item Show that $\Acal^{\sharp}$ is a smooth isomorphism $\Xcal^{s}(M)\to\Xcal^{s-2}(M)$.
	\hfill (\autoref{lem:A_pseudo_differential})
	\item Rewrite equation~\eqref{eq:maineq} as
	\begin{equation*}
		\frac{D}{\dd t}\frac{\dd\varphi}{\dd t} = \big((\Acal^\sharp)^{-1}S(\dot\varphi\circ\varphi^{-1})\big)\circ\varphi.
	\end{equation*}
	\item Show that $u\mapsto S(u)$ is a smooth map $\Xcal^{s}(M)\to\Xcal^{s-2}(M)$.
	\hfill  (\autoref{lem:diff_order_reduction})
	\item Show that 
	\begin{equation*}
		(\varphi,\dot\varphi)\mapsto \big((\Acal^\sharp)^{-1}S(\dot\varphi\circ\varphi^{-1})\big)\circ\varphi	
	\end{equation*}
	is smooth as a map $\Diff^{s}(M)\times H^{s}(M,TM) \to H^{s}(M,TM)$.
	\item Use standard results of smooth ordinary differential equations on Banach manifolds to obtain local existence and uniqueness.
\end{enumerate}
\end{framed}
}}

Consider the lifted inertia operator $\Acal^\sharp \coloneqq \sharp\circ \Acal$, with $\Acal$ defined in~\eqref{eq:Acal}.

\begin{lemma}\label{lem:A_pseudo_differential}
	$\Acal^{\sharp}$ is a smooth isomorphism $\Xcal^{s}(M) \to \Xcal^{s-2}(M)$.
\end{lemma}

\begin{proof}
	Let $R\colon\Omega^{1,s}(M)\to\Har^{1}(M)$ and $P_\mathrm{ex}\colon \Omega^{1,s}(M)\to \delta\Omega^{1,s+1}(M)$ be the Hodge projections onto $\Har^{1}(M)$ and $\delta\Omega^{1,s+1}(M)$ respectively.
	These mappings are smooth, as the Hodge decomposition of $\Omega^{1,s}(M)$ is smooth~\cite{EbMa1970}.
	Since the musical isomorphism is smooth it follows that $Z\colon u \mapsto (Ru^{\flat},P_{\mathrm{ex}}u^{\flat},\divv(u))$ is a smooth mapping $\Xcal^{s}(M) \to \Har^{1}(M)\times\delta\Omega^{1,s+1}(M)\times \Fcal^{s-1}(M)$;
	it is in fact an isomorphism, since it has a smooth inverse given by $(h,\sigma,\rho)\mapsto h^{\sharp} + \sigma^{\sharp} + \grad(\Delta^{-1}(\rho))$. 
	(Notice that $\Delta^{-1}\colon \Fcal^{s-1}(M)\to\Fcal^{s+1}(M)$.)

	From the definition \eqref{eq:Acal} of $\Acal$ it follows that
	\begin{equation*}
		\Acal^{\sharp} = Z^{-1}\circ (\id,(1-\gamma)\id - \alpha\Delta, -\beta\Delta )\circ Z
	\end{equation*}
	This is an isomorphism $\Xcal^{s}(M)\to\Xcal^{s-2}(M)$ since $((1-\gamma)\id-\alpha\Delta)\colon\delta\Omega^{1,s+1}(M) \to \delta\Omega^{1,s-1}(M)$ and $\Delta\colon\Fcal^{s-1}(M)\to\Fcal^{s-3}(M)$ are isomorphisms (see~\cite{Ta1996a}).
	The inverse is
	\begin{equation*}
		(\Acal^{\sharp})^{-1} = Z^{-1}\circ\big(\id,\big((1-\gamma)\id-\alpha\Delta\big)^{-1},-\frac{1}{\beta}\Delta^{-1} \big)\circ Z,
	\end{equation*}
	which concludes the proof.
\end{proof}

From the definition of $u$ it follows that $u(\varphi(x)) = \dot\varphi(x)$ for $x\in M$.
By differentiation with respect to $t$ we obtain
\begin{equation}\label{eq:second_derivation}
	\frac{\ud}{\ud t}\Big( u\big(\varphi(x)\big) \Big)  = \ddot\varphi(x) \in T^{2}_{(\varphi(x),\dot\varphi(x))}M .
\end{equation}
The Levi--Civita connection $\nabla$, induced by the Riemannian metric $\met$ on $M$, defines a diffeomorphism $(c,\dot c,\ddot c) \mapsto (c,\dot c,\nabla_{\dot c}\,\dot c)$ between the second tangent bundle $T^{2}M$ and the Whitney sum $TM\oplus TM$.
By pointwise operations, this diffeomorphism identifies the second tangent bundle $T^{2}\Diff^{s}(M)$ with the Whitney sum $T\Diff^{s}(M)\oplus T\Diff^{s}(M)$.
By the $\omega$--lemma (see, e.g.,~\cite{EbMa1970}) the identification is smooth.
Using this identification and $u = \dot\varphi\circ\varphi^{-1}$, we express equation~\eqref{eq:second_derivation} as
\begin{equation*}
	\dot u + \nabla_{u}u = \Big(\frac{D}{\ud t}\dot\varphi\Big)\circ\varphi^{-1},
\end{equation*}
where $\frac{D}{\ud t}\dot\varphi(x) \coloneqq \nabla_{\dot\varphi(x)}\dot\varphi(x)$ is the co-variant derivative along the path itself.
We now rewrite equation~\eqref{eq:maineq1} as
\begin{equation}\label{eq:pre_geodesic_spray}
	\Acal^\sharp \Big(\big(\frac{D}{\dd t}\frac{\dd\varphi}{\dd t}\big)\circ\varphi^{-1} \Big)
	= - (\LieD_u \Acal u)^\sharp - (\Acal^\sharp u) \divv(u) + \Acal^\sharp \nabla_u u
	\eqqcolon F(u) .
\end{equation}

The approach is to show that~\eqref{eq:pre_geodesic_spray} defines a smooth spray on $\Diff^{s}(M)$, i.e., a smooth vector field
\begin{equation*}
	\tilde S\colon T\Diff^{s}(M) \to T^{2}\Diff^{s}(M) \simeq T\Diff^{s}(M)\oplus T\Diff^{s}(M) .
\end{equation*}

Let $R_{\psi}\colon\Diff^{s}(M)\to\Diff^{s}(M)$ denote composition with $\psi\in\Diff^{s}(M)$ from the right, i.e., $R_{\psi}(\varphi) = \varphi\circ\psi$.
As already mentioned, this is a smooth mapping, so the corresponding tangent mapping $TR_{\psi}$, given by $T_{\varphi}\Diff^{s}(M) \ni v \mapsto v\circ\psi\in T_{\varphi\circ\psi}\Diff^{s}(M)$, is also smooth.
Let $T\Diff^{s-2}(M)\!\upharpoonright\!\Diff^{s}(M)$ denote the restriction of the tangent bundle $T\Diff^{s-2}(M)$ to the base $\Diff^{s}(M)$.
This is a smooth Banach vector bundle~\cite[Appendix~A]{EbMa1970}.
If $B\colon\Xcal^{s}(M)\to\Xcal^{s-2}(M)$ we denote by $\tilde B$ the bundle mapping $T\Diff^{s}(M)\to T\Diff^{s-2}(M)\!\upharpoonright\!\Diff^{s}(M)$ given by
\begin{equation*}
	\tilde B\colon (\varphi,\dot\varphi) \mapsto (\varphi,\tilde B_{\varphi}(\dot\varphi)), \quad \tilde B_{\varphi}(\dot\varphi)\coloneqq TR_{\varphi}\circ B\circ TR_{\varphi^{-1}} .
\end{equation*}

If $B$ is smooth, then $\tilde B_{\varphi}\colon T_{\varphi}\Diff^{s}(M)\to T_{\varphi}\Diff^{s-2}(M)$ is smooth for fixed $\varphi\in\Diff^{s}(M)$, but $\tilde B\colon T\Diff^{s}(M)\to T\Diff^{s-2}(M)$ need not be smooth, because $\varphi \mapsto \varphi^{-1}$ is not smooth.
The following lemmas resolve the situation in our specific case.

\begin{lemma}\label{lem:Atilde_smooth}
	The mapping
	\begin{equation*}
		\tilde{\Acal}^{\sharp}\colon T\Diff^{s}(M) \to T\Diff^{s-2}(M)\!\upharpoonright\!\Diff^{s}(M)
	\end{equation*}
	is a smooth vector bundle isomorphism.
\end{lemma}

\begin{proof}
	We have 
	\begin{equation*}
		\tilde{\Acal}^{\sharp} = \tilde R^{\sharp} + (1-\gamma)\tilde P_{\mathrm{ex}}^{\sharp} + \tilde W
	\end{equation*}
	where $P_{\mathrm{ex}}^{\sharp} = \sharp\circ P_{\mathrm{ex}}\circ\flat$, $R^{\sharp} = \sharp\circ R\circ \flat$ are the lifted Hodge projections and $W = \alpha\,\sharp\circ\delta\circ\dd\circ\flat - \beta\,\grad\circ\divv$.
	$\tilde P_{\mathrm{ex}}^{\sharp}$ and $\tilde R^{\sharp}$ are smooth bundle maps $T\Diff^{s}(M)\to T\Diff^{s}(M)$~\cite[Appendix A, Lemmas~2,3,6]{EbMa1970}.
	Thus, $\tilde P_{\mathrm{ex}}^{\sharp}$ and $\tilde R^{\sharp}$ are also smooth as mappings $T\Diff^{s}(M)\to T\Diff^{s-2}(M)\!\upharpoonright\!\Diff^{s}(M)$.
	That $\tilde W$ is smooth follows from~\cite[Appendix A, Lemma~2]{EbMa1970}.

	In a local chart in a neighbourhood of $\varphi\in\Diff^{s}(M)$, the derivative of $\tilde{\Acal}^{\sharp}$ at $(\varphi,\dot\varphi)$ is a smooth linear mapping of the form
	\begin{equation*}
		\begin{pmatrix}
			\id & 0 \\
			* & \tilde{\Acal}_{\varphi}^{\sharp}
		\end{pmatrix}\colon T_{\varphi}\Diff^{s}(M)\times T_{\varphi}\Diff^{s}(M) \to T_{\varphi}\Diff^{s}(M)\times T_{\varphi}\Diff^{s-2}(M)
	\end{equation*}
	It follows from \autoref{lem:A_pseudo_differential} that $\tilde{\Acal}^{\sharp}_{\varphi}$ is a linear isomorphism, with smooth inverse given by $\widetilde{(\Acal^{\sharp})^{-1}_{\varphi}}$.
	The result now follows from the inverse function theorem for Banach manifolds.
\end{proof}

\begin{lemma}\label{lem:differential_order_reduction}
	Let $B\colon\Xcal^{s}(M)\to\Xcal^{s-k}(M)$ be a smooth linear differential operator of order~$k$.
	If $s>n/2+k$, then
	\begin{equation*}
		u \longmapsto B\nabla_{u}u-\nabla_{u}B u = [B,\nabla_{u}]u 
	\end{equation*}
	is a smooth non-linear differential operator $\Xcal^{s}(M)\to\Xcal^{s-k}(M)$.
\end{lemma}

\begin{proof}
	If $f$ and $g$ are a scalar differential operators of order $k$ and $l$ respectively, then $[f,g]$ is a scalar differential operator of order $k+l-1$, since the order~$k+l$ differential terms in the commutator cancel each other.
	In general, this is not true for vector valued differential operators.
	Nevertheless, for a fixed $v$, the linear operator $u \mapsto \nabla_{v}u$ is given in components by
	\begin{equation*}
		\nabla_{v}u = \Big( v^{i}u^{j}\Gamma^{k}_{ij} + v^{i}\frac{\pd u^{k}}{\pd x^{i}} \Big) \mathbf{e}_{k},
	\end{equation*}
	so the differentiating part of $\nabla_{v}$ is acting diagonally on the elements of~$u$.
	We write $\nabla_{v}u = G u + f(u^{i})\mathbf{e}_{i}$, where $G\colon\Xcal^{s}(M)\to\Xcal^{s}(M)$ is tensorial and $f$ is a scalar differential operator of order~1.
	If $B = (b^{i}_{j})$, so that $B(u_{1}\pd_1 +\ldots + u_{n}\pd_{n}) = b^{i}_{j}(u^{j})\mathbf{e}_{i}$, then
	\begin{equation*}
		[B,\nabla_{v}]u = [B,G]u + \big([f,b^{i}_{j}]u^{j}\big) \mathbf{e}_{i}.
	\end{equation*}
	Since $G$ is tensorial, $[B,G]$ is a differential operator of the same order as $B$, i.e., order~$k$.
	Since $f$ and $b_{ij}$ are scalar differential operators of order $1$ and $k$, it holds that $[f,b_{ij}]$ is of order $k+1-1= k$.
	Since $\nabla_{v}B u$ differentiates $v$ zero times, and $B\nabla_{v}u$ differentiates $v$ at most $k$ times, it is now clear that the total operation $u\mapsto [B,\nabla_{u}u]$ differentiates $u$ at most~$k$ times, which finishes the proof.
	%
	%
	%
	%
	%
\end{proof}

\begin{lemma}\label{lem:conjugation_smooth}
	Let $B\colon\Xcal^{s}(M)\to\Xcal^{s-k}(M)$ be a smooth linear differential operator of order~$k$.
	If $s > n/2 + k$, then
	\begin{equation*}
		\tilde B\colon T\Diff^{s}(M) \to T\Diff^{s-k}(M)\!\upharpoonright\!\Diff^{s}(M)
	\end{equation*}
	is a smooth bundle map.
\end{lemma}

\begin{proof}
	Assume first that $B$ is of order~1.
	Locally, $\tilde B(\varphi,\dot\varphi)$ is then constructed by rational combinations of $\varphi^{i},\dot\varphi^{i},\frac{\ud \varphi^{i}}{\ud x^{j}},\frac{\ud \dot\varphi^{i}}{\ud x^{j}}$.
	Smoothness follows since pointwise multiplications are smooth operations~\cite[Appendix A, Lemma~2]{EbMa1970}.

	We now decompose $B$ into a sequence of~$k$ first order operators, so that $B = B_{1}\cdots B_{k}$.
	Then $\tilde B = \tilde B_{1} \cdots \tilde B_{k}$ and by the first part of the proof each $\tilde B_{i}$ is smooth and drops differentiability by one.
	This finishes the proof.
\end{proof}

\begin{remark}\label{rmk:bilinear_conjugation}
	Notice that if $Q\colon\Xcal^{s}(M)\times\Xcal^{s-l}(M)\to\Xcal^{s-k}$ is a bilinear differential operator of order~$k\geq 0$ in its first argument and $k-l\geq 0$ in its second argument, then \autoref{lem:conjugation_smooth} implies that
	\begin{equation*}
		\tilde Q\colon T\Diff^{s}(M)\times T\Diff^{s-l}(M)\!\upharpoonright\!\Diff^{s}(M) \to T\Diff^{s-k}(M)\!\upharpoonright\!\Diff^{s}(M)
	\end{equation*}
	is smooth whenever $s>n/2 + k$.
\end{remark}

\begin{lemma}\label{lem:S_tilde_smooth}
	If $s>n/2+2$, then
	\begin{equation*}
		\tilde F:T\Diff^{s}(M) \to T\Diff^{s-2}(M)\!\upharpoonright\!\Diff^{s}(M)
	\end{equation*}
	is a smooth bundle map.
\end{lemma}

\begin{proof}
	For any $v\in\Xcal(M)$ we have $(\LieD_u v^\flat)^\sharp = \LieD_u v + 2\mathrm{Def}(u)v$, where $\mathrm{Def}(u)$ is the type~$(1,1)$ tensor defined by $\frac{1}{2}(\LieD_u\met)(v,\cdot) = \met(\mathrm{Def}(u)v,\cdot)$ (see, e.g.,~\cite[\S\!~2.3]{Ta1996a}). 
	Thus
	\[
	\begin{split}
		F(u) &=
		- (\LieD_u \Acal u)^\sharp - (\Acal^\sharp u) \divv(u) + \Acal^\sharp \nabla_u u
		\\ &= -\LieD_u \Acal^\sharp u - 2\mathrm{Def}(u)\Acal^\sharp u -(\Acal^\sharp u) \divv(u)  + \Acal^\sharp \nabla_u u
		\\ &= -\nabla_u \Acal^\sharp u + \nabla_{\Acal^\sharp u} u - 2\mathrm{Def}(u)\Acal^\sharp u -(\Acal^\sharp u) \divv(u)  + \Acal^\sharp \nabla_u u
		\\ &= (\underbrace{\Acal^\sharp \nabla_u -\nabla_u \Acal^\sharp}_{[\Acal^\sharp,\nabla_u]}) u + \nabla_{\Acal^\sharp u} u - 2\mathrm{Def}(u)\Acal^\sharp u -(\Acal^\sharp u) \divv(u),
	\end{split}
	\]
	where we used that $\LieD_u v = \nabla_u v - \nabla_v u$.

	Let $Q\colon\Xcal^{s}(M)\times\Xcal^{s-2}(M)\to\Xcal^{s-2}(M)$ be the bilinear mapping
	\begin{equation*}
		Q(u,v) \coloneqq \nabla_{v} u - 2\mathrm{Def}(u)v - v \divv(u).
	\end{equation*}
	Notice that $Q$ is tensorial in $v$ and of order one in $u$.
	If $s>n/2+2$ then $Q$ is smooth.
	Write $\Acal^{\sharp} = P + W$, where $P = R^{\sharp} + (1-\gamma)P_{\mathrm{ex}}^{\sharp}$ and $W$ is a linear differential operator of order two as above.
	We now have
	\begin{equation*}
		F(u) = [\Acal^\sharp,\nabla_u]u + Q(u,\Acal^{\sharp}u) = [P,\nabla_{u}]u + [W,\nabla_{u}]u + Q(u,Pu) + Q(u,Wu).
	\end{equation*}
	The approach is to show that each of these terms are maximally of order two and smooth under conjugation with right translation.

	For the first term,
	\begin{equation*}
		\widetilde{[P,\nabla_{(\cdot)}]} = \tilde P \circ\widetilde{\nabla_{(\cdot)}} - \widetilde{\nabla_{(\cdot)}} \circ\tilde P.
	\end{equation*}
	We already know that $\tilde P\colon T\Diff^{s}(M)\to T\Diff^{s}(M)$ is smooth. 
	From \autoref{lem:conjugation_smooth} and \autoref{rmk:bilinear_conjugation} it follows that
	$\widetilde{\nabla_{(\cdot)}}\colon T\Diff^{s}(M) \to T\Diff^{s-1}(M)\!\upharpoonright\!\Diff^{s}(M)$ is smooth.

	For the second term, $(u,v)\mapsto [W,\nabla_{v}]u$ is a bilinear differential operator.
	From \autoref{lem:differential_order_reduction} it is of order two (since $W$ is of order two).
	From \autoref{lem:conjugation_smooth} and \autoref{rmk:bilinear_conjugation} it then follows that $\widetilde{[W,\nabla_{(\cdot)}]}$ is smooth.

	For the third term, it follows from \autoref{lem:conjugation_smooth} and \autoref{rmk:bilinear_conjugation} that $\tilde Q$ is smooth of order one. 
	Since $\tilde P$ is smooth of order zero, we get that $\widetilde{Q(\cdot,P\,\cdot\,)}$ is smooth of order one.

	For the fourth term, $(u,v)\mapsto Q(u,W v)$ is a bilinear differential operator of order one and two respectively in its arguments.
	It follows from \autoref{lem:conjugation_smooth} and \autoref{rmk:bilinear_conjugation} that $\widetilde{Q(\cdot,W\,\cdot\,)}$ is smooth of order two.

	Altogether, we have proved that $\tilde F\colon T\Diff^{s}(M) \to T\Diff^{s-2}(M)\!\upharpoonright\!\Diff^{s}(M)$ is smooth, which finishes the proof.
\end{proof}

Equation~\eqref{eq:pre_geodesic_spray} can be written
\begin{equation}\label{eq:geodesic_eq_TR}
	\tilde\Acal^{\sharp}\big(\varphi, \frac{D}{\dd t}\dot\varphi\big) = \tilde F(\varphi,\dot\varphi),
\end{equation}
from which we obtain the main result in this section.

\begin{theorem}\label{thm:smoothness_of_spray}
	If $s>n/2+2$, then the geodesic spray
	\begin{equation*}
		\tilde S\colon T\Diff^{s}(M)\ni (\varphi,\dot\varphi) \mapsto \Big(\varphi,\dot\varphi,\big((\tilde{\Acal}^{\sharp})^{-1}\circ \tilde F \big) (\varphi,\dot\varphi)\Big) \in T\Diff^{s}(M)\oplus T\Diff^{s}(M),
	\end{equation*}
	corresponding to the $\alpha$-$\beta$-$\gamma$~metric~\eqref{eq:metric}, is smooth.
\end{theorem}

\begin{proof}
	Follows from \autoref{lem:Atilde_smooth} and \autoref{lem:S_tilde_smooth}.
\end{proof}

This result implies local well-posedness and smooth dependence on initial data.

\begin{corollary}\label{cor:local_existence}
	Under the conditions in \autoref{thm:smoothness_of_spray}, the Riemannian exponential 
	\begin{equation*}
		\mathrm{Exp}\colon T\Diff^{s}(M) \to \Diff^{s}(M),
	\end{equation*}
	corresponding to the $\alpha$-$\beta$-$\gamma$~metric~\eqref{eq:metric}, is smooth.
	Further, if $\varphi\in\Diff^{s}(M)$, then
	\begin{equation*}
		\mathrm{Exp}_{\varphi}\colon T_{\varphi}\Diff^{s}(M) \to \Diff^{s}(M)
	\end{equation*}
	is a local diffeomorphism from a neighbourhood of the origin to a neighbourhood of $\varphi$.
\end{corollary}

\begin{proof}
	Follows from standard results about smooth sprays on Banach manifolds~\cite{La1999}.
\end{proof}


\section{Descending metrics and the space of densities} 
\label{sec:descending_metrics}

Let $\pi\colon E\to B$ be a smooth fibre bundle.
In this section we characterise pairs of metrics on~$E$ and~$B$ for which the projection~$\pi$ is a Riemannian submersion.
We do this in three steps, introducing more and more structure to the fibre bundle:
\begin{enumerate}
	\item First, the plain case $\pi\colon E\to B$.
	A characterisation of all descending metrics is given.
	\item Second, the case when $\pi\colon E\to B$ is a principal $H$--bundle.
	This allows a characterisation of descending metrics in terms of $H$--invariance.
	\item Third, the case when $E=G$, where $G$ is a Lie group, and $B$ is a $G$--homogeneous space, i.e., there is a transitive Lie group action of $G$ on $B$.
	Then the projection $\pi_b\colon g\mapsto b\cdot g$, for any fixed element $b\in B$, defines a principal $G_b$--bundle, where $G_b$ is the isotropy group of~$b$.
	This structure allows metrics on $G$ that are both right-invariant and descending.
\end{enumerate}

The main example is~$G=\Diff(M)$ and $B=\Dens(M)$, i.e, the space of smooth densities on~$M$ (see \eqref{eq:dens_def} below).
The main result is that the $\alpha$-$\beta$-$\gamma$~metric~\eqref{eq:metric} on~$\Diff(M)$ descends to the right-invariant canonical $L^2$~metric on~$\Dens(M)$, i.e., the Fisher metric.

Let us begin with the first case.
The kernel of the derivative of the projection map~$\pi$ defines the vertical distribution~$\mathcal{V}$ on $E$.
That is, for each $x\in E$
\begin{equation*}
	\mathcal{V}_x=\{ v\in T_xE \,;\, T_x\pi\cdot v=0 \}.
\end{equation*}
If $\met_E$ is a metric on~$E$, we can define the horizontal distribution~$\mathcal{H}=\mathcal{V}^{\bot}$ as the orthogonal complement of~$\mathcal{V}$ with respect to $\met_E$.

\begin{definition}\label{def:descending_metric}
	A metric $\met_E$ on $E$ is called \emph{descending} if there exists a metric $\met_B$ on $B$ such that
	\begin{equation*}
		\met_E(u,v) = (\pi^*\met_B)(u,v)\quad \forall\, u,v\in\mathcal{H}.
	\end{equation*}
\end{definition}

In other words, a metric on $E$ is descending if and only if there exists a metric on $B$ such that~$\pi$ is a \emph{Riemannian submersion}, i.e., such that $T\pi\colon TE\to TB$ preserves the length of horizontal vectors.

If $\met_E$ is descending, then $\met_B$ is unique, since $$T_x\pi\colon\mathcal{H}_x\to T_{\pi(x)}B$$ is an isomorphism for each~$x\in E$.

Now we show how to construct descending metrics.
Let $\met_B$ be a metric on~$B$.
Lift it to a positive semi-definite bilinear form~$\pi^*\met_B$ on~$E$.
Let~$\mathsf{h}$ be a positive semi-definite bilinear form on~$E$ such that~$\ker(\mathsf{h}) \cap \mathcal{V} = \{ 0 \}$ and the co-dimension of $\ker(\mathsf{h})$ is equal to the dimension of~$\mathcal{V}$.
Then
\begin{equation} \label{eq:constructed_metric}
	\met_E = \pi^*\met_B + \mathsf{h}
\end{equation}
is a descending metric on~$E$.
Notice that~$\ker(\pi^*\met_B) = \mathcal{V}$ and~$\ker(\mathsf{h}) = \mathcal{H}$.
Consequently, $\met_E(u,v) = \pi^*\met_B(u,v)$ for all $u,v\in\mathcal{H}$, so $\met_E$ is indeed descending.
Also notice that the horizontal distribution is independent of the choice of~$\met_B$.

The form~\eqref{eq:constructed_metric} characterises all descending metrics.
Indeed, if $\met_E$ is a descending metric, let $\met_B$ be the corresponding metric on $B$, and let $P\colon TE\to \mathcal{V}$ be the orthogonal projection onto $\mathcal{V}$ with respect to~$\met_E$.
Then $\met_E$ is of the form~\eqref{eq:constructed_metric} with $\mathsf{h}(u,v)\coloneqq\met_E(u,Pv)$.

Consider now the second case.
That is, let $H$ be a Lie group and consider the case when $\pi\colon E\to B$ is a principal $H$--bundle, with a left action $L_h\colon E\to E$ for $h\in H$.
Being a principal bundle, the fibres are parameterised by~$H$, so $\pi\circ L_h=\pi$ and if $\pi(x)=\pi(y)$ then there exists a unique~$h\in H$ such that $y =  L_h(x)$.
Thus, if $\met_B$ is a metric on $B$, then $\pi^*\met_B = (\pi\circ L_h)^*\met_B = L_h^*\pi^*\met_B$.
Thus, if $\met_E$ is descending, then 
\begin{equation}\label{eq:invariance}
	(L_h^*\met_E)(u,v) = \met_E(u,v)\quad\forall u,v\in \mathcal{H}.
\end{equation}
The converse is also true.

\begin{proposition}\label{pro:descending_principal_bundle}
	Let $\met_E$ be a metric on~$E$. Then $\met_E$ is descending if and only if it fulfils~\eqref{eq:invariance}.
\end{proposition}

\begin{proof}
	We have shown ``$\Rightarrow$'', so ``$\Leftarrow$'' remains.
	Assuming~\eqref{eq:invariance}, define~$\met_B$ as follows.
	For $\bar u,\bar v \in T_{\pi(x)}B$, take any point $y\in\pi^{-1}(\{x\})$. 
	The linear map $T_y\pi\colon\mathcal{H}_y\to T_{\pi(x)}B$ is an isomorphism, so we get $u,v\in \mathcal{H}_y$ by $u=T_y\pi^{-1}\cdot\bar u$ and $v=T_y\pi^{-1}\cdot\bar v$.
	Define $\met_B$ by
	\begin{equation*}
		\met_B(\bar u,\bar v) \coloneqq \met_E(u,v) .
	\end{equation*}
	This is a well-defined metric on $\met_E$, i.e., it is independent on which $y\in\pi^{-1}(\{x\})$ we use.
	Indeed, for another $y'\in\pi^{-1}(\{ x\})$ we get $u',v' \in \mathcal{H}_{y'}$ as above. 
	Also, $y' = L_h(y)$ for some~$h\in H$.
	From~\eqref{eq:invariance} it follows that $\met_E(u,v) = \met_E(u',v')$.
	By construction, $\met_E(u,v) = (\pi^*\met_B)(u,v)$ for all $u,v\in\mathcal{H}$, so~$\met_E$ is indeed descending.
\end{proof}

We now specialise further, to the third case.
Let $G$ be a Lie group with identity~$e$.
Denote by $L_g$ and $R_g$ respectively the left and right action of $g\in G$ on~$G$.
Assume that $G$ has a right transitive action $\bar R_g$ on a manifold~$B$.
($B$ is then called a $G$--homogeneous space.)
If $b\in B$, then $G_b = \{ g\in G; \bar R_g(b)=b \}$ denotes the isotropy Lie subgroup of~$G$.
For every $b\in B$ we obtain a principal $G_b$--bundle $\pi_b\colon G\to B$, with $\pi_b(g) \coloneqq \bar R_g(b)$. 
This structure implies that $B$ is diffeomorphic to the homogeneous space $G_b\backslash G$ of right co-sets.
The map $\pi_b$, which is well-defined on $G_b\backslash G$, provides a diffeomorphism.
(If $b,b'\in B$, then $G_b$ and $G_{b'}$ are conjugate subgroups, i.e., there exists $g\in G$ such that $g G_b g^{-1} = G_{b'}$.)

We are interested in right-invariant metrics on $G$, i.e., metrics $\met_G$ that fulfil
\begin{equation*}
	\met_G(u,v) = \met_G(TR_g\cdot u, TR_g\cdot u),
\end{equation*}
or, equivalently, $R_g^*\met_G = \met_G$.
Being right-invariant does \emph{not} imply being descending.
For a right-invariant metric $\met_G$ to be descending (with respect to $\pi_b$), it follows from \autoref{pro:descending_principal_bundle} that $\met_G$ must fulfil
\begin{equation*}
	L_h^*R_{g}^*\met_G(u,v) = \met_G(u,v)\quad \forall u,v\in\mathcal{H}^b ,\quad\forall g\in G,\quad \forall h\in G_b,
\end{equation*}
where $\mathcal{H}^b$ denotes the horizontal distribution.
Since $\met_G$ is right-invariant and since the right action descends to $G_b\backslash G$, i.e., $R_g$ maps fibres to fibres, it is enough to check the condition for $g=h^{-1}$ and for vectors $u,v\in \mathcal{H}^{b}_e = \mathfrak{g}_b^{\bot}$, where $\mathfrak{g}_b$ is the Lie algebra of~$G_b$.
Indeed, the following result is given in~\cite{KhLeMiPr2013}.

\begin{proposition}\label{pro:descending_right_invariant}
	Let $\met_G$ be a right-invariant metric on~$G$.
	Then $\met_G$ is descending (with respect to $\pi_b$) if and only if
	\begin{equation*}
		\met_G(\ad_\xi(u),v) + \met_G(u,\ad_\xi(v)) = 0
		\quad \forall \, u,v\in\mathfrak{g}_b^\bot, \; \xi\in\mathfrak{g}_b.
	\end{equation*}
\end{proposition}

Consider now the reverse question: If $\met_G = \pi^*\met_B + \mathsf{h}$ is descending, when is it right-invariant?
Since right invariance means that $\met_G = R_g^*\met_G$ it must hold that
$R_g^*\pi^*\met_B = \pi^*\met_B$ and~$R_g^*\mathsf{h} = \mathsf{h}$.
Also, since $$R_g^*\pi^*\met_{B} = (\pi\circ R_g)^*\met_{B} = (\bar R_g\circ\pi)^*\met_B = \pi^*\bar R_g^*\met_B$$ we obtain the following result.

\begin{proposition}\label{pro:right_invariant_descending}
	Let $\met_G = \pi^*\met_{B} + \mathsf{h}$ be a descending metric on~$G$.
	Then~$\met_G$ is right-invariant if and only if both~$\met_B$ and~$\mathsf{h}$ are right-invariant, i.e.,
	\begin{equation*}
		\bar R_g^*\met_B = \met_B 
		\qquad\text{and}\qquad
		R_g^*\mathsf{h} = \mathsf{h}
	\end{equation*}
	for all $g\in G$.
\end{proposition}

\iftoggle{final}{}{
{\color{blue!50!black}
One may ask what it means for a descending metric to be right-invariant, in terms of the fibre structures generated by the transitive action.
The result is that right-invariant and descending with respect to $\pi_b$ is equivalent to being descending with respect to $\pi_{b'}$ for any $b'\in B$.

\begin{proposition}\label{pro:right_conjugate_descending}
	A metric $\met_G$ is right-invariant and descending with respect to $\pi_b$ if and only if it is descending with respect to $\pi_{b'}$ for any $b'\in B$.
\end{proposition}

\begin{proof}
	...
\end{proof}
}}

Let us now investigate our case of interest, i.e., $G=\Diff(M)$.
First, consider the manifold of smooth probability densities on $M$, which takes the rôle of~$B$ above.
It is given by
\begin{equation}\label{eq:dens_def}
	\Dens(M) = \left\{ \nu\in\Omega^n(M) \, ;\, \nu>0, \int_M\nu = 1 \right\}.
\end{equation}
The tangent spaces of $\Dens(M)$ are $T_\nu\Dens(M)=\Omega^{n}_0(M) \coloneqq \{ a\in\Omega^{n}(M); \int_M a = 0 \}$.
$\Diff(M)$ acts on $\Dens(M)$ from the right by pullback $\bar R_\varphi(\nu) = \varphi^{*}\nu$.
The corresponding lifted action is again given by pullback, i.e., $T\bar R_{\varphi}(\nu,a) = (\varphi^{*}\nu,\varphi^{*}a)$.

Recall the volume form $\vol\in\Dens(M)$ corresponding to the Riemannian structure on~$M$.
Using $\bar R_\varphi$ and $\vol$, we define the projection map $\pi_\vol\colon\Diff(M)\to\Dens(M)$ by $\pi_\vol(\varphi) = \bar R_\varphi(\vol) = \varphi^{*}\vol$.
\citet{Mo1965} proved that~$\pi_{\vol}$ is surjective, which implies that $\pi_{\vol}$ is a submersion (see \autoref{rem:dens_bundle_structure}).
The corresponding isotropy group is $\Diffvol(M)$, i.e., 
\begin{equation*}
	\pi_\vol(\psi\circ\varphi) = \pi_\vol(\varphi) \quad\text{if and only if}\quad \psi\in\Diffvol(M).
\end{equation*}
Accordingly, with $\Diffvol(M)$ acting on $\Diff(M)$ from the left, we have a principal bundle structure
\begin{equation}\label{eq:diff_principal_bundle}
	\Diffvol(M) \xhookrightarrow{\quad} \Diff(M) \xrightarrow{\;\pi_\vol\;} \Dens(M) .
\end{equation}
The vertical distribution~$\mathcal{V}$ of this bundle structure is given by vectors in $T\Diff(M)$ that are divergence-free when right translated to~$T_\id\Diff(M) = \Xcal(M)$.
That is,
\begin{equation*}
	\mathcal{V}_\varphi = \{ v\in T_\varphi\Diff(M) \, ;\, v\circ\varphi^{-1} \in \Xcalvol(M) \} .
\end{equation*}

Relative to the abstract formulation presented above, $G=\Diff(M)$, $B=\Dens(M)$, and $G_b = \Diffvol(M)$.
In particular, 
$
\Diffvol(M)\backslash\Diff(M) \simeq \Dens(M)
$.

\begin{remark}\label{rem:dens_bundle_structure}
	The analysis involved in the bundle structure~\eqref{eq:diff_principal_bundle} can be made precise in the category of Sobolev manifolds.
	Indeed, 
	\[
		\Diff^s_\vol(M)\backslash\Diff^s(M) \simeq \Dens^{s-1}(M) \quad\text{if}\quad s>n/2+1,
	\]
	where the projection $\pi_{\vol}\colon\Diff^s(M)\to\Dens^{s-1}(M)$ is smooth.
	These results can be used to prove that $\pi_{\vol}\colon\Diff(M)\to\Dens(M)$ is smooth with respect to an inverse limit Hilbert (ILH) topology.
	Notice that the principal bundle structure 
	\begin{equation*}
		\Diff^{s}_{\vol}(M) \xhookrightarrow{\quad} \Diff^{s}(M) \xrightarrow{\quad} \Diff^s_\vol(M)\backslash\Diff^s(M)
	\end{equation*}
	is only $C^{0}$, since the left action of $\Diff^{s}(M)$ on itself is only continuous.
	See Ebin and Marsden~\cite[\S\!~5]{EbMa1970} for details.
	By using the Nash--Moser inverse function theorem, Hamilton~\cite[\S\!~III.2.5]{Ha1982} showed directly that $\pi_{\vol}\colon\Diff(M)\to\Dens(M)$ is a smooth principal $\Diffvol(M)$--bundle with respect to a Fréchet topology.
\end{remark}

{

\begin{remark}\label{rem:reference_density_ambiguity}
	The choice of reference element $\vol\in\Dens(M)$ is not canonical; 
	it specifies which point we consider to be the ``identity density''.
	If $\nu\in\Dens(M)$ is another density, then $\Diffvol(M)$ and $\Diff_\nu(M)$ are conjugate subgroups:
	$\Diff_\nu(M) = \psi^{-1}\circ\Diffvol(M)\circ\psi$ for any $\psi\in \{ \varphi\in\Diff(M); \varphi^{*}\vol = \nu \}$.
\end{remark}

We now turn to right-invariant and descending metrics on $\Diff(M)$.
There is a natural $L^2$~metric on $\Dens(M)$ given by
\begin{equation}\label{eq:canonical_metric_on_dens}
	\inner{a,b}_{\nu} = \int_M \frac{a}{\nu}\frac{b}{\nu} \, \nu, \qquad a,b\in \Omega^{n}_{0}(M), 
\end{equation}
where $a/\nu \in \Omega^{0}(M)$ is defined by $a = (a/\nu)\nu$ and correspondingly for $b/\nu$. (Equivalently, $a/\nu$ is the Radon--Nikodym derivative of the measure induced by $a$ with respect to the measure induced by $\nu$.)
The metric~\eqref{eq:canonical_metric_on_dens} is called the Fisher metric.
As already mentioned, it is fundamental in the theory of information geometry.
In addition, it is used in statistical mechanics for measuring ``thermodynamic length''~\cite{Cr2007,FeCr2009}, and in the geometrical formulation of quantum mechanics~\cite{FaKuMaMaSuVe2010}.
The Fisher metric is canonical: it is independent of the Riemannian structure on~$M$.
 
\begin{remark}
One can also write the Fisher metric \eqref{eq:canonical_metric_on_dens} as
\[
	\inner{a,b}_{\nu} = \int_M (\star_{\nu}a) b,
\]
where $\star_{\nu}:\Omega^{n}(M)\to\Omega^{0}(M)$ is the Hodge star on $n$--forms corresponding to~$\nu$.
\end{remark}

\begin{proposition}\label{pro:invariance_of_density_metric}
	The Fisher metric \eqref{eq:canonical_metric_on_dens} on $\Dens(M)$ is invariant with respect to the action $\bar R_\varphi$.
\end{proposition}

\begin{proof}
	\begin{equation*}
		\inner{\varphi^{*}a,\varphi^{*}b}_{\varphi^{*}\nu} 
		= \int_M (\star_{\varphi^*\nu}\varphi^*a) \varphi^*b 
		= \int_M \varphi^*\big((\star_{\nu}a) b\big)
		= \inner{a,b}_{\nu}.
	\end{equation*}
\end{proof}

We now come to the main result of this section.

\begin{theorem}\label{thm:metric_is_descending}
	The $\alpha$-$\beta$-$\gamma$~metric~\eqref{eq:metric} on $\Diff(M)$ descends to a metric on $\Dens(M)$, that, up to multiplication with $\beta$, is the Fisher metric~\eqref{eq:canonical_metric_on_dens}.
\end{theorem}

\begin{proof}
	First, since the inner product~\eqref{eq:inner_product} on $\Xcal(M)$ corresponding to the $\alpha$-$\beta$-$\gamma$~metric preserves orthogonality with respect to the Helmholtz decomposition, it follows that the horizontal distribution is given by
	\begin{equation*}
		\mathcal{H}_\varphi = \big\{ v\in T_\varphi\Diff(M) \, ;\, v\circ\varphi^{-1} \in \grad(\Fcal(M)) \big\} ,
	\end{equation*}
	i.e., vectors that, when translated to the identity, are given by gradient vector fields.
	If $u\in \mathcal{V}_\id = \Xcalvol(M)$ and $f,g\in\Fcal(M)$, then
	\begin{equation*}
		\begin{split}
		\pair{\LieD_u \grad(f),\grad(g)}_{\alpha\beta\gamma} 
		&= \int_M \beta\,\delta (\LieD_u \grad(f))^{\flat}\delta\grad(g)^{\flat}\, \vol \\
		&= \int_M \beta\,\dd \interior_{\LieD_u \grad(f)}\vol \wedge\star\,\dd\interior_{\grad(g)}\vol \\
		&= -\int_M \beta\,\dd \interior_{\grad(f)}\vol \wedge\star\,\dd\interior_{\LieD_u\grad(g)}\vol \\
		&= -\pair{\grad(f),\LieD_u\grad(g)}_{\alpha\beta\gamma},
		\end{split}
	\end{equation*}
	where we have used $\LieD_u\vol =0$ and $\LieD_u\star a = \star \LieD_u a$ for any $a\in\Omega^{n}(M)$.
	From \autoref{pro:descending_right_invariant} it follows that the $\alpha$-$\beta$-$\gamma$~metric is descending.

	The tangent map $T\pi_\vol$ restricted to $T_\id\Diff(M)=\Xcal(M)$ is given by $u\mapsto \LieD_u\vol$.
	Now,
	\begin{equation*}
		\begin{split}
			\pair{\grad(f),\grad(g)}_{\alpha\beta\gamma} &= 
			\int_M \beta\,\interior_{\grad(f)}\vol\wedge\star\interior_{\grad(g)}\vol \\
			&= \int_M \beta\,\LieD_{\grad(f)}\vol\wedge\star\LieD_{\grad(f)}\vol \\
			&= \beta\inner{\LieD_{\grad(f)}\vol,\LieD_{\grad(g)}\vol}_{\vol}.
		\end{split}
	\end{equation*}
	Therefore, the $\alpha$-$\beta$-$\gamma$~metric for horizontal vectors at the identity tangent space is given by the Fisher metric multiplied by $\beta$ of the projection of the horizontal vectors to the tangent space $T_\vol\Dens(M)$.
	Since this holds at one tangent space, it follows from \autoref{pro:invariance_of_density_metric} that it holds at every tangent space (both the $\alpha$-$\beta$-$\gamma$~metric and the Fisher metric are right-invariant).
	This concludes the proof.
\end{proof}

A consequence of \autoref{thm:metric_is_descending} is a geometric explanation of the observation in \autoref{sub:hodge_components}, that solutions that are initially gradients remain gradients.
Indeed, \citet{He1960} proved that initially horizontal geodesics remain horizontal for any descending metric.

\begin{remark}
	The ``components'' $\met_B$ and $\mathsf{h}$ of the $\alpha$-$\beta$-$\gamma$~metric are identified as follows:
	\begin{equation*}
		\pair{u,v}_{\alpha\beta\gamma} 
		= \underbrace{\pair{\bar P_{\gamma} u^{\flat},\bar P_{\gamma} v^{\flat}}_{L^{2}} +\alpha\pair{\dd u^{\flat},\dd v^{\flat}}_{L^{2}} }_{\mathsf{h}}+ 
		\underbrace{\beta\pair{\delta u^{\flat},\delta v^{\flat}}_{L^{2}} }_{\pi^*\met_B},
	\end{equation*}
	with the same notation as in equation~\eqref{eq:inner_product}.
\end{remark}

\section{Optimal transport and factorisation} 
\label{sec:optimal_transport_and_polar_factorisation}

Optimal mass transport has a long history, going back to \citet{Mo1781} and Kantorovich~\cite{Ka1942,Ka1948}.
For a modern treatise, see the lecture notes by \citet{Ev2001_lecture_notes}, \citet{AmGi2009}, or \citet{Mc2010}, or the monograph by \citet{Vi2009}.

In this section we study the relation between ``distance square'' optimal transport problems and descending metrics.
Typically, optimal transport problems are formulated with minimal regularity restrictions.
In contrast, we consider smooth formulations (more precisely Sobolev~$H^{s}$ with $s>n/2+1$).
We take the point of view of \citet[\S~4]{Ot2001}, but $\Dens(M)$ is identified with right co-sets instead of left.
A central topic is the correspondence between optimal transport problems and polar factorisations.
The main result is given in~\autoref{sub:optimal_info_trans}, where we establish existence and uniqueness results for optimal information transport, and a matching polar factorisation of~$\Diff^{s}(M)$.
As a finite-dimensional analogue, we show in \autoref{sub:qr} that $QR$~factorisation of square matrices can be seen as polar factorisation corresponding to optimal transport of inner products.

Abstract geometric optimal transport problems are formulated as follows.
Let $G$ be a Lie group acting transitively on a manifold~$B$.
Assume $G$ is equipped with a cost function $c\colon G\times G\to \R^{+}$.
This renders a geometric formulation of Monge's original problem:
\begin{equation}\label{eq:optimal_transport_abstract}
\text{
Given $b,b'\in B$, find $T\in \{g\in G; g\cdot b = b' \}$ minimising $c(e,T)$.
}
\end{equation}

We are interested in the case $c = \dist_G^{2}$, where $\dist_G$ is the geodesic distance of a metric~$\met_G$.
In particular, we are interested in the case when $\met_G$ is descending with respect to a fibre bundle structure $\pi\colon G\to B$.
Then, the optimal transport problem~\eqref{eq:optimal_transport_abstract} is simplifies to: (i) find the shortest curve on $B$ connecting $b$ and $b'$, (ii) lift that curve to a horizontal curve on~$G$, and (iii) take the endpoint as the solution.
This simplification is possible because a shortest curve between two fibres must be horizontal:

\begin{lemma}\label{lem:horiz_shortest_curve}
	Let $\pi\colon G\to B$ be a Riemannian submersion, and let $\zeta\colon [0,1]\to G$ be an arbitrary smooth curve.
	Then there exists a unique horizontal curve $\zeta_h\colon [0,1]\to G$ such that $\zeta_h(0)=\zeta(0)$ and $\pi\circ\zeta=\pi\circ\zeta_h$.
	The length of~$\zeta_h$ is less than or equal the length of~$\zeta$, with equality if and only if $\zeta$ is horizontal.
\end{lemma}

\begin{proof}
	For each $t\in [0,1]$ there is a unique decomposition $\dot\zeta(t) = v(t) + h(t)$, where $v(t)\in\mathcal{V}_{\zeta(t)}$ and $h(t)\in\mathcal{H}_{\zeta(t)}$.
	Thus, we have the curves $t\mapsto v(t) \in \mathcal{V}$ and $t\mapsto h(t)\in\mathcal{H}$.
	By the projection $\pi$ we also get a curve $\bar\zeta(t) = \pi(\zeta(t)) \in B$.
	This curve can be lifted to a horizontal curve as follows.
	Take any time-dependent smooth vector field~$\bar X_t$ on~$B$ for which $\bar\zeta$ is an integral curve, i.e., $\dot{\bar\zeta}(t) = \bar X_t(\bar\zeta(t))$.
	Now lift $\bar X_t$ to its corresponding horizontal section $X_t(g) = (T_g\pi)^{-1}\cdot \bar X_t(\pi(g))$.
	(We can do this since $T_g\pi\colon\mathcal{H}_g\to T_{\pi(g)}B$ is an isomorphism.)
	Next, let $\zeta_h$ be the unique integral curve of $X_t$ with $\zeta_h(0) = \zeta(0)$.
	By construction, $\pi(\zeta_h(t)) = \pi(\zeta(t))$ and $\met_G(\dot\zeta_h(t),\dot\zeta_h(t)) = \met_G(h(t),h(t))$.
	Thus, $\met_G(\dot\zeta_h(t),\dot\zeta_h(t))\leq \met_G(\dot\zeta(t),\dot\zeta(t))$, with equality if and only if $\dot\zeta(t)\in\mathcal{H}$.

	It remains to show that $\zeta_h$ is unique.
	Assume $\zeta_h':[0,1]\to G$ is another horizontal curve with $\pi\circ\zeta_h' = \bar\zeta$ and $\zeta_h'(0) = \zeta(0)$.
	By differentiation with respect to $t$ we obtain
	\begin{equation*}
		T_{\zeta_h'(t)}\pi\cdot \dot\zeta_h'(t) = \dot{\bar\zeta}(t) = \bar X_t(\bar\zeta(t)) = 
		\bar X_t(\pi(\zeta_h'(t))) = T_{\zeta_h'(t)}\pi\cdot X_t(\zeta_h'(t)).
	\end{equation*}
	Since $\zeta_h'$ is horizontal, it is an integral curve of $X_t$.
	Since $\zeta_h'$ and $\zeta_h$ have the same initial conditions, uniqueness of integral curves yield $\zeta_h' = \zeta_h$.
	This concludes the proof.
\end{proof}

\begin{remark}
	\autoref{lem:horiz_shortest_curve} is independent of the group structure of $G$.
	Indeed, $G$ and $B$ can be Banach manifolds with a smooth fibre bundle structure $\pi\colon G\to B$.
	The metrics can be weak, which is important for the main example in \autoref{sub:optimal_info_trans}.
\end{remark}

\iftoggle{final}{}{
{\color{blue!50!black}
\begin{framed}
	If $(B,\dist_B)$ is a complete metric space, then the infimum of shortest curves (occurring in the definition of $\dist_B$) is attained, and so there exists a corresponding horizontal curve in $G$, which is a shortest curve between $g$ and $\pi^{-1}(\{ b'\})$.
	In our main example, where $G=\Diff(M)$ and $B=\Dens(M)$, it does \emph{not} hold that $(B,\dist_B)$ is complete.
	However, there are still shortest curves (even unique minimal geodesics), since $\Dens(M)$ is a convex subset of a complete metric space.
\end{framed}
}}

For cost functions given by the square distance of a descending metric, \autoref{lem:horiz_shortest_curve} reduces the optimal transport problem~\eqref{eq:optimal_transport_abstract} to a problem entirely on~$B$: find a shortest curve between two given elements $b'',b'\in B$.
This problem is not always easier to solve.
If, however, the geometry of the Riemannian manifold $(B,\met_B)$ is well understood, for example if any two elements in $B$ are connected by a minimal geodesic, then problem~\eqref{eq:optimal_transport_abstract} simplifies significantly (see \autoref{pro:abstract_optimal_transport} below).

The concept of \emph{polar factorisation} is strongly related to optimal transport.
Following \citet{Br1991}, we introduce the \emph{polar cone}
\begin{equation}\label{eq:polar_cone}
	K = \{ k\in G; \dist_G(e,k)\leq\dist_G(h,k),\;\forall h\in G_b \}.
\end{equation}
Expressed in words, the polar cone~\eqref{eq:polar_cone} consists of elements in $G$ whose closest point on the identity fibre is~$e$.

\begin{proposition}\label{pro:abstract_optimal_transport}
	Let $\met_G$ be a right-invariant metric on $G$ that descends to a metric $\met_B$ on $B$ with respect to the fibre bundle $\pi_b(g) = \bar R_g(b)$ for some fixed $b\in B$.
	Then the following statements are equivalent:
	\begin{enumerate}
		\item For any $b'\in B$, there exists a unique minimal geodesic from $b$ to $b'$.
		\item For any $b'',b'\in B$, there exists a unique minimal geodesic from $b''$ to $b'$.
		\item There exists a unique solution to the optimal transport problem~\eqref{eq:optimal_transport_abstract} with $c=\dist_{G}^{2}$, and that solution is connected to~$e$ by a unique minimal geodesic.
		\item Every $g\in G$ has a unique factorisation $g = h k$, with $h\in G_b$ and $k\in K$, and every $k\in K$ is connected to~$e$ by a unique minimal geodesic.
	\end{enumerate}
\end{proposition}

\begin{proof}
	$1\Rightarrow 2$. 
	Since $\met_G$ is right-invariant we get that if $\zeta\colon[0,1]\to G$ is a minimal geodesic, then so is $\bar R_g\circ\zeta$ for any $g\in G$.
	Since the action $\bar R$ is transitive, $b'' = \bar R_g(b)$ for some $g\in G$.

	$2\Rightarrow 3$. 
	Let $\bar\zeta\colon[0,1]\to B$ be the minimal geodesic from $b$ to $b'$.
	By \autoref{lem:horiz_shortest_curve}, there is a unique corresponding horizontal geodesic $\zeta\colon[0,1]\to G$ with $\zeta(0) = e$ and $\pi_b(\zeta(t))=\bar\zeta(t)$.
	There cannot be any curve from $e$ to $\pi_b^{-1}(\{ b' \})$ shorter than $\zeta$, because then $\bar\zeta$ would not be a minimal geodesic.
	A curve from $e$ to $\pi_b^{-1}(\{ b' \})$ of the same length as $\zeta$ must be horizontal (by \autoref{lem:horiz_shortest_curve}), and therefore equal to~$\zeta$ (also by \autoref{lem:horiz_shortest_curve}).
	Thus, if $q \in \pi_b^{-1}(\{ b' \})\backslash\{\zeta(1)\}$ then $\dist_G(e,q)>\dist_G(e,\zeta(1))$, so $\zeta(1)$ is the unique solution to problem~\eqref{eq:optimal_transport_abstract}.
	Also, $\zeta$ is a unique minimal geodesic between $e$ and $\zeta(1)$.

	$3\Rightarrow 4$.
	Let $k$ be the unique solution to~\eqref{eq:optimal_transport_abstract} with $b' = \pi_b(g)$.
	Then $k$ and $g$ belong to the same fibre, so $g=h k$ for some unique element $h\in G_b$.
	There cannot be another such factorisation $g=h' k'$, because then $k$ would not be a unique solution to~\eqref{eq:optimal_transport_abstract}.
	Now take any $k\in K$.
	Then $k$ is the unique solution to~\eqref{eq:optimal_transport_abstract} with $b'= \pi_b(k)$, and that solution is connected to~$e$ by a unique minimal geodesic.
	Thus, any $k\in K$ is connected to~$e$ by a unique minimal geodesic.

	$4\Rightarrow 1$.
	For any $b'\in B$ we can find $g\in G$ such that $b' = \bar R_g(b) = \pi_b(g)$, which follows since the action $\bar R$ is transitive.
	Let $g=hk$ be the unique factorisation, and let $\zeta\colon[0,1]\to G$ be the unique minimal geodesic from $e$ to $k$.
	Assume now that $\zeta$ is not horizontal.
	Then, by \autoref{lem:horiz_shortest_curve}, we can find a horizontal curve $\zeta_h\colon[0,1]\to G$ with $\zeta_h(0)=k$ and $\zeta_h(1)\in G_b$ that is strictly shorter than~$\zeta$.
	Since $\zeta$ is a unique minimal geodesic between $e$ and $k$, it cannot hold that $\zeta_h(1) = e$.
	Thus, $k\notin K$ since there is a point $\zeta_h(1)$ on the identity fibre closer to $k$ than $e$, so we reach a contradiction.
	Therefore, $\zeta$ must be horizontal, so it descends to a corresponding geodesic $\bar\zeta$ between $b$ and $b'$.
	$\bar\zeta$ must be unique minimal, otherwise $\zeta$ cannot be unique minimal.
	This finalises the proof.
\end{proof}

\begin{remark}
	If the metric $\met_B$ is not right-invariant, then \autoref{pro:abstract_optimal_transport} is still valid in the case $b''=b$.
\end{remark}

\subsection{Optimal information transport} 
\label{sub:optimal_info_trans}

We now return to the example of main interest, namely $G=\Diff^{s}(M)$ and $B=\Dens^{s-1}(M)$.
Recall \autoref{thm:metric_is_descending}, that the $\alpha$-$\beta$-$\gamma$~metric on $\Diff^{s}(M)$ descends to the Fisher metric on $\Dens^{s-1}(M)$ up to multiplication with~$\beta$.
For simplicity, we assume $\beta=1$ throughout this section.
Also recall
\begin{itemize}
	\item $\bar R_\varphi(\nu) = \varphi^{*}\nu$, so $\pi_\vol(\varphi) = \varphi^{*}\vol$;
	\item $\Diff^{s}(M)$ and $\Dens^{s-1}(M)$ are Banach manifolds if $s> n/2+1$;
	\item the projection $\pi_\vol\colon\Diff^{s}(M)\to\Dens^{s-1}(M)$ is smooth.
\end{itemize}
Thus, all prerequisites in \autoref{pro:abstract_optimal_transport} are fulfilled.

\citet*{KhLeMiPr2013} showed that the geodesic boundary-value problem on $\Dens^{s-1}(M)$, with respect to the Fisher metric, can be formulated as an optimal transport problem, with respect to a degenerate cost function.
Since the cost function is degenerate, solutions are not unique, so there is no corresponding polarisation result.
The $\alpha$-$\beta$-$\gamma$~metric on $\Diff^{s-1}(M)$ allows us to obtain a non-degenerate optimal transport formulation in accordance with the framework in~\autoref{sec:optimal_transport_and_polar_factorisation}. 
In particular, we obtain a factorisation result for diffeomorphisms.

Let $\lambda,\nu\in\Dens^{s-1}(M)$.
We consider the following optimal transport problem: 
\begin{equation}\label{eq:optimal_info_trans}
	\text{Find $\varphi\in\{ \eta\in\Diff^{s}(M);\eta^{*}\lambda=\nu\}$ minimising $\dist^{2}_{\alpha\beta\gamma}(\id,\varphi)$.}
\end{equation}
Here, $\dist_{\alpha\beta\gamma}$ is the Riemannian distance corresponding to the $\alpha$-$\beta$-$\gamma$~metric~\eqref{eq:metric}.
Since the $\alpha$-$\beta$-$\gamma$~metric descends to Fisher's information metric, we refer to problem~\eqref{eq:optimal_info_trans} as \emph{optimal information transport};
find the most optimal diffeomorphism pulling the probability density~$\lambda$ to~$\nu$.
For simplicity, we assume that $\lambda = \vol$.


As mentioned in the introduction, \citet{Fr1991} showed that the Fisher metric has constant curvature.
Therefore, its geodesics are easy to analyse.
Indeed, following~\cite{KhLeMiPr2013}, we introduce the infinite-dimensional sphere of radius $r = \sqrt{\vol(M)}$
\begin{equation*}
	S^{\infty,s}(M) = \Big\{
		f\in H^{s}(M,\R); \pair{f,f}_{L^{2}}=\vol(M)
	\Big\}.
\end{equation*}
If $s>n/2$, then $S^{\infty,s}(M)$ is a Banach submanifold of $H^{s}(M,\R)$.
The $L^{2}$ inner product on $H^{s}(M,\R)$, restricted to $S^{\infty,s}(M)$, provides a weak Riemannian metric.
Although weak, the geodesic spray is smooth.
The geodesics are given by great circles, so $S^{\infty,s}(M)$ is geodesically complete and its diameter is given by $\pi \sqrt{\vol(M)}$.

Let $\mathcal{O}^{s}(M) = \{f\in S^{\infty,s}(M); f>0 \}$ denote the set of positive functions of radius~$\sqrt{\vol(M)}$.
$\mathcal{O}^{s}(M)$ is an open subset of $S^{\infty,s}(M)$ and therefore a Banach manifold itself.
The following result is given in~\cite{KhLeMiPr2013}.

\begin{theorem}\label{thm:sphere_diffeo}
	If $s>n/2$, then the map
	\begin{equation*}
		\Phi\colon\Dens^{s}(M)\ni \nu \longmapsto \sqrt{\frac{\nu}{\vol}}
	\end{equation*}
	is an isometric diffeomorphism $\Dens^{s}(M)\to \mathcal{O}^{s}(M)$.
	The diameter of $\mathcal{O}^{s}(M)$, thus $\Dens^{s}(M)$, is $\frac{\pi}{2}\sqrt{\vol(M)}$.
\end{theorem}

If $f,g\in\mathcal{O}^{s}(M)$, then there is a unique minimal geodesic $\sigma\colon[0,1]\to S^{\infty,s}(M)$ from $f$ to $g$, and that geodesic is contained in $\mathcal{O}^{s}(M)$, i.e., $\mathcal{O}^{s}(M)$ is a convex subset of $S^{\infty,s}(M)$.
Indeed, the minimal geodesic is given by
\begin{equation}\label{eq:explicit_geodesic}
	\sigma\colon[0,1]\ni t \longmapsto \frac{\sin\!\big((1-t)\theta\big)}{\sin \theta}f + \frac{\sin(t\theta)}{\sin \theta}g,
	\quad \theta = \arccos\Big( \frac{\pair{f,g}_{L^{2}}}{\vol(M)} \Big) .
\end{equation}

The polar cone of $\Diff^{s}(M)$ with respect to the $\alpha$-$\beta$-$\gamma$~metric is
\begin{equation}\label{eq:polar_cone_diff}
	K^{s}(M) = \big\{
		\varphi\in\Diff^{s}(M); \dist_{\alpha\beta\gamma}(\id,\varphi) \leq \dist_{\alpha\beta\gamma}(\eta,\varphi),\;
		\forall\, \eta\in\Diff_{\vol}^{s}(M)
	\big\}.
\end{equation}
There is a unique minimal geodesic between $\vol$ and any $\nu\in\Dens^{s-1}(M)$, so from \autoref{pro:abstract_optimal_transport} every $\psi\in K^{s}(M)$ is the endpoint of a minimal horizontal geodesic $\zeta\colon[0,1]\to\Diff^{s}(M)$ with $\zeta(0) = \id$.
Since $\zeta$ is horizontal, it is of the form $\zeta(t)=\Exp_{\id}(t\grad(w_0))$ for a unique $w_0\in\Fcal^{s+1}(M)$, where $\Exp\colon T\Diff^{s}(M)\to\Diff^{s}(M)$ denotes the Riemannian exponential corresponding to the $\alpha$-$\beta$-$\gamma$~metric.

Let $\varphi\in\Diff^{s}(M)$.
Due to the explicit form~\eqref{eq:explicit_geodesic} of minimal geodesics in $\mathcal{O}^{s-1}(M)$, thus $\Dens^{s-1}(M)$, we can compute the function $w_0\in\Fcal^{s+1}(M)$ such that $\Exp_\id(\grad(w_0))$ is the unique element in $K^{s}(M)$ belonging to the same fibre as $\varphi$.
Indeed,
\begin{equation*}
	\begin{split}
	T_{\id}\pi_\vol\cdot\grad(w_0) &= \frac{\dd}{\dd t}\Big|_{t=0} \pi_\vol(\Exp_\id(t\grad(w_0))) \\
	&= \frac{\dd}{\dd t}\Big|_{t=0}\pi_\vol(\zeta(t)) = \frac{\dd}{\dd t}\Big|_{t=0}\sigma(t)^{2}\vol 
	= 2\dot\sigma(0)\sigma(0)\vol,		
	\end{split}
\end{equation*}
where $\sigma(t)$ is the curve \eqref{eq:explicit_geodesic} with $f=1$ and $g=\sqrt{\Jac(\varphi)}$ (the Jacobian is defined by $\Jac(\psi)\vol = \psi^{*}\vol$).
Since $\sigma(0) = 1$ and $T_{\id}\pi_\vol\cdot\grad(w_0) = \LieD_{\grad(w_0)}\vol = \Delta w_0\vol$, we get
\begin{equation}\label{eq:w0}
	\Delta w_0 = \frac{2\theta\sqrt{\Jac(\varphi)}-2\theta\cos(\theta)}{\sin\theta}, \quad \theta = \arccos\Big(\frac{\int_M \sqrt{\Jac(\varphi)} \vol}{\vol(M)}\Big).
\end{equation}

How do we compute the horizontal geodesic $\zeta(t) = \Exp_\id(t \grad(w_0))$?

One way is to solve equation~\eqref{eq:maineq} with $u(0) = \grad(w_0)$, and reconstruct~$\zeta(t)$ by integrating the non-autonomous equation $\dot\zeta(t) = u(t)\circ\zeta(t)$ with $\zeta(0)=\id$.
From \cite[Proposition~4.5]{KhLeMiPr2013} and the isomorphism $T_\id\pi_\vol\colon \grad(\Fcal^{s+1}(M))\to T_\vol\Dens^{s-1}(M)$, the solution exists for $t\in [0,1]$ (for details on the maximal existence time, see \cite[\S\!~4.2]{KhLeMiPr2013}).

Another way is to directly lift the geodesic curve~\eqref{eq:explicit_geodesic} by the technique in the proof of \autoref{lem:horiz_shortest_curve}.
Indeed, the geodesic $\bar\zeta(t)$ in $\Dens^{s-1}(M)$ corresponding to $\zeta(t)$ is given by $\bar\zeta(t) = \Phi^{-1}(\sigma(t)) = \sigma(t)^{2}\vol$.
Since $\pi_\vol(\zeta(t)) = \bar\zeta(t)$ and 
\begin{equation*}
	TR_{\zeta(t)^{-1}}(\zeta(t),\dot\zeta(t)) = (\id,\dot\zeta(t)\circ\zeta(t)^{-1}) = (\id,\grad(w_t)), \quad w_t\in\Fcal^{s+1}(M)	
\end{equation*}
we conclude
\begin{equation*}
	\begin{split}
	\dot{\bar\zeta}(t) &= T_\id (\bar R_{\zeta(t)}\circ \pi_\vol) \cdot \grad(w_t) \\
		&= \zeta(t)^{*}\big(\LieD_{\grad(w_t)}\vol \big) \\
		&= \divv(\grad(w_t))\circ\zeta(t)\Jac(\zeta(t))\vol.
	\end{split}
\end{equation*}
From $\dot{\bar\zeta}(t) = 2\dot\sigma(t)\sigma(t)\vol$ and $\Jac(\zeta(t)) = \sigma(t)^{2}$, we get
\begin{equation*}
	2\dot\sigma(t)\sigma(t) = \big(\Delta w_t\circ\zeta(t)\big) \sigma(t)^{2} . 
\end{equation*}
The horizontal geodesic $\zeta(t)$ is now constructed by solving the following non-autonomous ordinary differential equation on $\Diff^{s}(M)$
\begin{equation}\label{eq:lifting}
	\begin{split}
		\dot\zeta(t) &= \grad(w_t)\circ\zeta(t), \quad \zeta(0)=\id \\
		\Delta w_t &= \frac{2\dot\sigma(t)}{\sigma(t)}\circ\zeta(t)^{-1} \\
		\sigma(t) &= \frac{\sin\!\big((1-t)\theta\big)}{\sin \theta} + \frac{\sin(t\theta)}{\sin \theta}\sqrt{\Jac(\varphi)},
		\quad \theta = \arccos\Big(\frac{\int_M \sqrt{\Jac(\varphi)} \vol}{\vol(M)}\Big).
	\end{split}
\end{equation}
Explicitly, the vector field $X_{t}$ on $\Diff^{s}(M)$ is
\begin{equation*}
	X_t\colon\zeta \mapsto \widetilde{\grad}\circ \tilde{\Delta}^{-1}\Big( \zeta, \frac{2\dot\sigma(t)}{\sigma(t)} \Big) = \grad\Big(\Delta^{-1}\Big( \frac{2\dot\sigma(t)}{\sigma(t)}\circ\zeta^{-1} \Big)\Big)\circ \zeta .
\end{equation*}
That is, equation~\eqref{eq:lifting} is the non-autonomous ordinary differential equation
\begin{equation*}
	\dot\zeta(t) = X_t(\zeta(t)),\quad \zeta(0) = \id .
\end{equation*}
Smoothness of $X_{t}$ is obtained by the same techniques as in the proofs of \autoref{lem:Atilde_smooth} and \autoref{lem:conjugation_smooth}.

In summary, we have proved the following result.


\begin{theorem}\label{thm:factorisation_of_diff}
	Let $s>n/2+1$.
	Every $\varphi\in\Diff^{s}(M)$ admits a unique factorisation $\varphi = \eta\circ\psi$, with $\eta\in\Diff^{s}_{\vol}(M)$ and $\psi \in K^{s}(M)$.
	We have $\psi = \Exp_\id(\grad(w_0))$ with $w_0$ given by equation~\eqref{eq:w0}.
	There is a unique minimal geodesic $\zeta(t)$ with $\zeta(0)=\id$ and $\zeta(1)=\psi$;
	it can be computed by solving equation~\eqref{eq:lifting}.
	The geodesic $\zeta(t)$ is horizontal.
\end{theorem}

\begin{remark}
	The factorisation in \autoref{thm:factorisation_of_diff} is independent of~$\alpha$, $\beta$, and $\gamma$, because each parameter choice yields the same horizontal distribution and horizontal geodesics.
\end{remark}




\subsection{Optimal transport of inner products and $QR$~factorisation} 
\label{sub:qr}

In this section we show how the $QR$~factorisation of square matrices is related to optimal transport of inner products on~$\R^{n}$. 
The example provides a finite-dimensional analogue of optimal information transport described in~\autoref{sub:optimal_info_trans}.
We do not address questions of global existence and uniqueness of geodesics (local existence and uniqueness follows automatically).
The aim is to provide geometric insight into the $QR$ and Cholesky factorisations.

Let $G=\GL(n)$ over the field $\R$.
Let $B=\Sympn$ be the manifold of inner products on $\R^{n}$.
$\Sympn$ is identified with the space of symmetric positive definite $n\times n$~matrices;
if $M$ is a symmetric positive definite matrix, then its corresponding inner product is $\pair{\mathbf{x},\mathbf{y}}_{M} = \mathbf{x}^{\trans}M\mathbf{y}$.
$\Sympn$ is a convex open subset of the vector space $\Symn$ of all symmetric $n\times n$~matrices.

The group $\GL(n)$ acts on $\Sympn$ from the right by $\bar R_{A}(M) = A^{\trans}M A$.
This action is transitive.
The lifted action is $T_M\bar R_{A}\cdot U = A^{\trans}U A$, where $U\in T_M\Sympn=\Symn$.

Let $I$ denote the identity matrix (which is an element in both $\GL(n)$ and $\Sympn$). 
Consider the projection $\pi_I\colon\GL(n)\to\Sympn$ given by $\pi_I(A) = \bar R_A(I) = A^{\trans}A$.
The corresponding isotropy group is $G_I = \SO(n)$, because
\begin{equation*}
\pi_I(Q A) = A^{\trans}Q^{\trans}Q A = A^{\trans}A = \pi_I(A)
\end{equation*}
for all $Q\in\SO(n)$.	
Thus, we have a principal bundle
\begin{equation}\label{eq:principal_bundle_QR}
	\SO(n) \xhookrightarrow{\quad} \GL(n) \xrightarrow{\;\pi_I\;} \Sympn .
\end{equation}

There is natural metric $\met_B$ on $\Sympn$ given by
\begin{equation}\label{eq:sym_metric}
	\met_{B,M}(U,V) = \tr(UM^{-2} V), \quad U,V\in T_M\Sympn .
\end{equation}
This metric is invariant with respect to the $\bar R_A$ action.
Indeed,
\begin{equation*}
	\begin{split}
		\met_{B,\bar R_A(M)}(T_M\bar R_A\cdot U, T_M\bar R_A\cdot V) &= \tr(A^{\trans}UA (A^{\trans}MA)^{-2} A^{\trans}VA) \\
		&= \tr((A^{\trans}MA)^{-1}A^{\trans}UA (A^{\trans}MA)^{-1} A^{\trans}VA) \\
		&= \tr(A^{-1}M^{-1}A^{-\trans}A^{\trans}UA A^{-1}M^{-1}A^{-\trans} A^{\trans}VA) \\
		&= \tr(A^{-1}M^{-1}UM^{-1}VA) \\
		& \text{(using cyclic property: $\tr(ABC) = \tr(BCA)$)} \\
		&= \tr(M^{-1}UM^{-1}VA A^{-1}) \\
		&= \tr(M^{-1}UM^{-1}V) \\
		&= \tr(UM^{-2}V) = \met_{B,M}(U,V)
	\end{split}
\end{equation*}

We proceed by constructing a metric on $\GL(n)$.
Consider the projection operator $\ell\colon\Matn\to\Matn$ given by 
\begin{equation*}
	\ell(U)_{ij} = \begin{cases}
		U_{ij} &\text{if } i \geq j, \\
		0 &\text{otherwise.}
	\end{cases}	
\end{equation*}
In words, $\ell(U)$ is zero on the strictly upper triangular entries and equal to~$U$ elsewhere.
Let $\met_{G}$ be the right-invariant metric on $\GL(n)$ defined by
\begin{equation}\label{eq:QR_metric}
	\met_{G,I}(u,v) = \tr\big(\ell(u)^{\trans}\ell(v)\big) + \tr\big((u+u^{\trans})(v+v^{\trans})\big).
\end{equation}
By right translation~$\met_{G,A}(U,V) = \met_{G,I}(UA^{-1},VA^{-1})$.
The orthogonal complement of $T_I\SO(n) = \so(n)$ with respect to $\met_{G,A}$ consists of the upper triangular matrices.
This follows since matrices in $\so(n)$ are skew symmetric, so the second term in~\eqref{eq:QR_metric} vanishes if either $u$ or $v$ belong to $\so(n)$.
In mathematical terms
\begin{equation*}
	\so(n)^{\trans} = \mathfrak{upp}(n) \coloneqq \{ u\in \gl(n); \ell(u) = 0 \}.	
\end{equation*}

\begin{proposition}
	The right-invariant metric $\met_G$ on $\GL(n)$ is descending with respect to the principal bundle structure~\eqref{eq:principal_bundle_QR}.
	The corresponding metric on $\Sympn$ is $\met_B$ in \eqref{eq:sym_metric}.
\end{proposition}

\begin{proof}
	By \autoref{pro:descending_right_invariant} we need to show that 
	\begin{equation*}
		\met_{G,I}( \ad_\xi(u),v) + \met_{G,I}(u,\ad_\xi(v)) = 0, \quad \forall\, u,v\in \mathfrak{upp}(n), \xi\in\so(n).
	\end{equation*}
	We have
	\begin{equation*}
		\met_{G,I}(\ad_\xi(u),v) = \tr\big( ([\xi,u]+[\xi,u]^{\trans})(v+v^{\trans})\big) = 
		\tr\big(([\xi,u+u^{\trans}])(v+v^{\trans}) \big).
	\end{equation*}
	By the cyclic property of the trace
	\begin{equation*}
		\begin{split}
			\tr\big(([\xi,u+u^{\trans}])(v+v^{\trans}) \big) &= -\tr\big((u+u^{\trans})([\xi,v+v^{\trans}]) \big) \\
			&= -\tr\big((u+u^{\trans})([\xi,v] + [\xi,v]^{\trans}]) \big) \\
			&= -\met_{G,I}(u,\ad_\xi(v)).
		\end{split}
	\end{equation*}
	Therefore, the metric is descending.
	If $u\in\mathfrak{upp}(n)$, then $T_I\pi_I\cdot u = u + u^{\trans}$, so
	\begin{equation*}
		\met_{G,I}(u,v) = \frac{1}{4}\tr\big((u+u^{\trans})(v+v^{\trans}) \big) = \met_{B,I}(T_I\pi_I\cdot u,T_I\pi_I\cdot v).
	\end{equation*}
	Since $\met_{B}$ is right-invariant, $\met_G$ descends to $\met_B$.
	This proves the result.
\end{proof}

The horizontal distribution $\mathcal{H}$ is given by $\mathcal{H}_A = \mathfrak{upp}(n)A$.
Since $\mathfrak{upp}(n)$ is a Lie algebra, i.e., closed under the matrix commutator, the horizontal distribution is integrable.
Its integral manifold through the identity is the Lie group $\mathup{Upp}(n)$ of upper triangular $n\times n$~matrices with strictly positive diagonal entries.
Notice that $\mathup{Upp}(n)$ forms the polar cone~\eqref{eq:polar_cone}.

Let $A\in\GL(n)$.
If there exists a unique minimal geodesic $\bar\zeta\colon[0,1]\to\Sympn$ from $I$ to $\pi_I(A)$, then, by \autoref{pro:abstract_optimal_transport}, we obtain a factorisation $A=QR$, with $Q\in\SO(n)$ and $R\in\mathup{Upp}(n)$.
Since the metric~\eqref{eq:QR_metric} is smooth, it follows from standard results in Riemannian geometry that there exists a neighbourhood $\mathcal{O}\subset \Sympn$ of $I$ such that any element in $\mathcal{O}$ is connected to $I$ by a unique minimal geodesic.
Therefore, $A$ has a unique $QR$~factorisation if $\pi_I(A)$ is close enough to $I$.

\begin{remark}	
	The $QR$~factorisation of any matrix $A$ is well-known to exist.
	It is unique if $A$ is invertible, which suggests unique minimal geodesics.
	Details of these questions are not investigated in this paper.
\end{remark}

In summary, the factor~$R$ in the $QR$ factorisation of $A$ solves the problem of optimally (with respect to the cost function $\dist_G^{2}$) transporting the Euclidean inner product on $\R^{n}$ to the inner product defined by $M=A^{\trans}A$.
In addition, the factor~$R$ is the transpose of the Cholesky factorisation of~$M$.
Indeed, if $L=R^{\trans}$ then 
\begin{equation*}
	M=\pi_I(A) = \pi_I(R) = R^{\trans}R = LL^{\trans}.	
\end{equation*}

Usually, the $QR$~factorisation is obtain by direct linear algebraic manipulations.
A different approach is to solve the geodesic equation on $\Sympn$, and lift the geodesic to a horizontal geodesic on $\GL(n)$.
Although probably inefficient compared to existing algorithms (there are very fast algorithms based on Householder reflections), the geodesic approach may provide insight, for example in the case of sparse matrices.

\begin{remark}
	The setting can be extended to $\GL(n,\C)$ by replacing $\SO(n)$ with $\mathup{U}(n)$, and every transpose with the Hermitian conjugate.
\end{remark}

\iftoggle{final}{}{
{\color{blue!50!black}

\begin{remark}
	Consider the case $G=\Diff(\mathbb{T}^{n})$ of diffeomorphisms on the flat $n$--torus.
	$\varphi\in \Diff(\mathbb{T}^{n})$ acts on $H^{s}(\mathbb{T}^{n},\R)$ by
	\begin{equation*}
		A_\varphi\colon H^{s}(\mathbb{T}^{n},\R) \to H^{s}(\mathbb{T}^{n},\R), \quad
		A f \coloneqq f\circ \varphi .
	\end{equation*}

\end{remark}

}}

\begin{remark}
	\citet{Ot2001} studied gradient flows with respect to the Wasserstein metric in optimal mass transport (``Otto calculus'', cf.~\cite{Vi2009}). 
	Analogous gradient flows for the optimal transport problem in this section are related to the Toda flow, the ``$QR$~algorithm''~\cite{GoLo1989}, and Brockett's flow for continuous diagonalisation of matrices~\cite{Br1988}.
\end{remark}


\iftoggle{final}{}{
{\color{blue!50!black}

\subsection{Optimal mass transport ($L^{2}$)} 
\label{sub:optimal_mass_transport}

\todo[inline]{Add lots of references here, and discuss history a bit.}

Let $(M,\met)$ be a compact Riemannian manifold as above.
Consider the $L^{2}$ metric on $H^{s}(M,M)$ given by
\begin{equation}\label{eq:L2_optimal_transport_metric}
	\inner{u,v}_{\varphi} = \int_M \met(u(x),v(x))\vol(x).
\end{equation}
Notice that this metric is invariant with respect to the right action of the isotropy group $\Diff_\vol^{s}(M)$ on $H^{s}(M,M)$.
Also, this $L^{2}$~metric induces a metric on the open subset $\Diff^{s}(M)\subset H^{s}(M,M)$.
We take $G=\Diff^{s}(M)$.

As above, $\Diff^{s}(M)$ acts on $B=\Dens^{s-1}(M)$ from the right by $\bar R_\varphi(\nu) = \varphi^{*}\nu$.
The Monge problem is: given $\nu,\nu'\in\Dens^{s-1}(M)$, find $\varphi\in \{\phi\in \Diff^{s}(M); \phi^{*}\nu = \nu' \}$ such that $\dist_{L^{2}}(\id,\varphi)$ is minimized.
Here, $\dist_{L^{2}}(\id,\varphi)$ is the length of the shortest curve from $\id$ to $\varphi$ with respect to the metric~\eqref{eq:L2_optimal_transport_metric}.
By switching the rôle of $\nu$ and $\nu'$ we can reformulate this problem using the pushforward instead of the pullback: find $\varphi\in \{\phi\in \Diff^{s}(M); \phi_{*}\nu = \nu' \}$ such that $\dist_{L^{2}}(\id,\varphi)$ is minimized.
Now, define the projection $\pi:\Diff^{s}(M)\to \Dens^{s-1}(M)$ by $\pi(\varphi) = \varphi_*\vol = \bar R_{\varphi^{-1}}(\vol)$.
The fibre of $\varphi\in\Diff^{s}(M)$ is given by $\varphi\circ\Diff_{\vol}^{s}(M)$, i.e., we consider left co-sets.

Since the metric~\eqref{eq:L2_optimal_transport_metric} is right-invariant, and since we consider left co-sets, it follows automatically that it is a descending metric with respect to~$\pi$.
The corresponding metric on $\Dens^{s-1}(M)$ is the \emph{Wasserstein metric}, given by
\begin{equation*}
	\inner{a,b}_{\nu} = \int_M \met\big(\grad (f),\grad (g)\big) \nu, 
	\quad \LieD_{\grad (f)}\nu = a, 
	\quad \LieD_{\grad (g)}\nu = b.
\end{equation*}
Notice that computing the metric requires the solution of the Poisson type problem to obtain $f,g\in\Fcal(M)$.
Indeed, if $\nu = \rho\,\vol$, then the equation for $f$ is 
\begin{equation*}
	\divv\big(\rho\grad(f)\big) = \frac{\dd a}{\dd \vol}	
\end{equation*}
and correspondingly for $g$.


}}


\section*{Acknowledgements} 

I would like to thank Martins Bruveris, Darryl Holm, Boris Khesin, Martin Kohlmann, Gerard Misio{\l}ek} and Matthew Perlmutter for very helpful discussions.

The work is supported by the \href{http://www.kva.se}{Royal Swedish Academy of Science} and the Swedish Research Council (contract {\small \href{http://vrproj.vr.se/detail.asp?arendeid=89639}{VR--2012--335}}).



\bibliographystyle{amsplainnat} 
\bibliography{tgvolume}


\end{document}